\documentclass[onecolumn,12pt]{IEEEtran} 

\usepackage{amsfonts,color,morefloats}
\usepackage{amssymb,amsthm, amsmath,latexsym}

\newtheorem{theorem}{Theorem}
\newtheorem{lemma}[theorem]{Lemma}

\newtheorem{corollary}[theorem]{Corollary}

\newtheorem{example}[theorem]{Example}

\newtheorem{question}[theorem]{Question}

\newcommand\myatop[2]{\genfrac{}{}{0pt}{}{#1}{#2}}

\newcommand{\ord}{{\mathrm{ord}}}

\newcommand{\lcm}{{\mathrm{lcm}}}
\newcommand{\tr}{{\mathrm{Tr}}}

\newcommand{\gf}{{\mathrm{GF}}}
\newcommand{\PG}{{\mathrm{PG}}}

\newcommand{\support}{{\mathrm{suppt}}}
\newcommand{\Aut}{{\mathrm{Aut}}}
\newcommand{\PAut}{{\mathrm{PAut}}}
\newcommand{\MAut}{{\mathrm{MAut}}}
\newcommand{\GAut}{{\mathrm{Aut}}}

\newcommand{\wt}{{\mathtt{wt}}}

\newcommand{\Z}{\mathbb{{Z}}}

\newcommand{\m}{\mathbb{M}}

\newcommand{\cP}{{\mathcal{P}}}
\newcommand{\cB}{{\mathcal{B}}}
\newcommand{\cA}{{\mathcal{A}}}

\newcommand{\C}{{\mathcal{C}}}
\newcommand{\M}{{\mathsf{M}}}

\newcommand{\Sim}{{\mathcal{S}}}

\newcommand{\cH}{{\mathcal{H}}}
\newcommand{\cR}{{\mathcal{R}}}

\newcommand{\ba}{{\mathbf{a}}}

\newcommand{\bc}{{\mathbf{c}}}

\newcommand{\bv}{{\mathbf{v}}}

\newcommand{\bzero}{{\mathbf{0}}}
\newcommand{\bone}{{\mathbf{1}}}

\newcommand{\bD}{{\mathbb{D}}}

\newcommand{\p}{{\mathcal P}}
\newcommand{\B}{{\mathcal B}}
\newcommand{\R}{{\mathcal R}}
\newcommand{\V}{{\mathcal V}}
\newcommand{\T}{{\mathcal T}}

\newcommand{\GA}{{\mathrm{GA}}}

\begin{document}

\title{The minimum linear locality of linear codes}

\author{Pan Tan,
        Cuiling Fan,
       Cunsheng Ding
        and~Zhengchun Zhou
\thanks{P. Tan, C. Fan, and Z. Zhou are with the School of Mathematics, Southwest Jiaotong University, Chengdu, 611756, China. E-mail: lanqingfeixue@my.swjtu.edu.cn, fcl@swjtu.edu.cn, zzc@swjtu.edu.cn.}
\thanks{C. Ding is with the Department of Computer Science
                           and Engineering, The Hong Kong University of Science and Technology,                                                 Clear Water Bay, Kowloon, Hong Kong, China (email: cding@ust.hk)}
}
\maketitle

\begin{abstract}
Locally recoverable codes (LRCs) were proposed for the recovery of data in distributed and cloud storage systems about nine
years ago. A lot of progress on the study of LRCs has been made by now. However, there is a lack of general theory on the minimum linear locality of linear codes. In addition, the minimum linear locality of many known families of linear codes is not studied in the literature. Motivated by these two facts, this paper develops some general theory about the minimum linear locality of linear codes, and investigates the minimum linear
locality of a number of families of linear codes, such as $q$-ary Hamming codes, $q$-ary Simplex codes, generalized Reed-Muller codes, ovoid codes, maximum arc codes, the extended hyperoval codes, and near MDS codes. Many classes of both distance-optimal and dimension-optimal LRCs are presented in this paper. The minimum linear locality of many families of linear codes are settled with the general theory developed in this paper.
\end{abstract}

\begin{IEEEkeywords}
Cyclic code, \and linear code, \and locally recoverable code, \and near MDS code, \and punctured code, \and shortened code.
\end{IEEEkeywords}

\section{Introduction of motivations, objectives and methodology}\label{sec-MOM}

Throughout this paper, let $n$ be a positive integer and let $q$ be a prime power.
An $[n, k, d]$ code $\C$ over $\gf(q)$ is a $k$-dimensional subspace of $\gf(q)^n$ with Hamming distance $d$.
We use $A_i(\C)$ or $A_i$, $\dim(\C)$, $d(\C)$ and $\C^\perp$ to denote the number of codewords of Hamming
weight $i$ in $\C$, the dimension of $\C$,  the minimum Hamming distance of $\C$, and the dual of $\C$.
The weight distribution and weight enumerator of $\C$ are defined by the sequence $(A_0, \ldots, A_n)$ and
the polynomial $\sum_{i=0}^n A_i z^i$, respectively. $\C$ is said to be a $t$-weight code if the sequence
$(A_1, \ldots, A_n)$ has Hamming weight $t$.

Denote $[n]=\{0,1,\ldots,n-1\}$ for each positive integer $n$. We usually index the coordinates of the codewords in $\C$ with the elements in $[n]$.
An $[n, k, d]$ code $\C$ over $\gf(q)$ is called an
$(n, k, d,q; r)$-LRC (locally recoverable code) if for each $i \in [n]$ there is a subset $R_i \subseteq [n] \setminus \{i\}$
of size $r$ and a function $f_i(x_1, \ldots, x_r)$ on $\gf(q)^r$ such that $c_i=f_i(\bc_{R_i})$ for each codeword
$\bc=(c_0, \ldots, c_{n-1})$ in $\C$, where $\bc_{R_i}$ is the projection of $\bc$ at $R_i$. The symbol $c_i$ is called
the $i$-th \emph{code symbol} and the set $R_i$ is called the \emph{repair set} or \emph{recovering set} of the
code symbol $c_i$. In this definition of LRCs, the degrees of the functions $f_i$ are not restricted. If we require that
each $f_i$ be a homogeneous function of degree $1$ in the definition above, then we say that $\C$ is $(n, k, d, q; r)$-LLRC
(linearly local recoverable code) and has linear locality $r$. By definition, a code $\C$ has locality $r$ if it has linear locality $r$.
But the converse may not be true. If a linear code has locality, it must have the minimum locality. The same is true for linear locality.  Regarding linear locality, we have the following questions.

\begin{question}\label{quest-bas1}
What linear codes have linear locality?
\end{question}

\begin{question}\label{quest-bas2}
If a linear code has linear locality, what is the minimum linear locality and how does one compute the minimum linear locality?
\end{question}

The first objective of this paper is to answer the two questions above. We will develop some general theory answering these
two questions.

For any $(n, k, d, q; r)$-LLRC, Gopalan et al. proved the following upper bound on the minimum distance $d$ \cite{GHSY12}:
\begin{equation}\label{eq-bound}
  d\leq n-k-\left\lceil\frac{k}r\right\rceil+2.
\end{equation}
The bound in \eqref{eq-bound} is similar to the Singleton bound, so we call it the Singleton-like bound. If an $(n,k,d,q;r)$-LLRC meets the Singleton-like bound with equality, then we say that the $(n,k,d,q;r)$-LLRC is distance-optimal ($d$-optimal for short).
If an $(n,k,d,q;r)$-LLRC meets the Singleton-like bound minus one with equality, then we say that the $(n,k,d,q;r)$-LLRC is almost distance-optimal (almost $d$-optimal for short).
 Note that the Singleton-like bound is not tight for codes over small finite
fields, as it is independent of the alphabet size $q$.

For any $(n, k, d, q; r)$-LLRC, Cadambe and Mazumdar developed the following bound on the dimension $k$ \cite{CM13}, \cite{CM15}:
\begin{equation}\label{cm-bound}
k\leq\min_{t\in \Z_+}[tr+k_{opt}^{(q)}(n-t(r+1),d)],
\end{equation}
where $\Z_+$ denotes the set of all positive integers, and
$k_{opt}^{(q)}(n,d)$ is the largest possible dimension of a linear code with length $n$, minimum distance $d$, and alphabet size $q$.  In this paper, we call the bound in (\ref{cm-bound}) the CM bound. An $(n,k,d,q;r)$-LLRC that attains the
CM bound with equality is said to be dimension-optimal ($k$-optimal for short).

While constructing new optimal LLRCs is an important task, searching for optimal LLRCs in the known families of linear
codes is also important.
The second objective of this paper is to study the minimum linear locality of certain known families of linear codes
and try to find out $d$-optimal or $k$-optimal LLRCs. We focus on non-binary linear codes, as the linear locality of
some families of binary
codes were studied in \cite{HYUS20}.
Our methodology is combinatorial and group-theoretical.

Locally recoverable codes were proposed for the recovery of data in distributed and cloud storage systems by
Gopalan, Huang, Simitci and Yikhanin \cite{GHSY12}. In the past nine years, a lot of progress on the study
of locally recoverable codes has been made. The reader is referred to
\cite{CCFT19,CFMST20,CH2019,CM13,CM15,GHSY12,HYUS20,JKZ20,LMT20,LXY19,Micheli20,TB14,TZSP20,WZL15,XY19}
and the references therein for information. Despite of the good progress made by now, the two questions raised above
look still open, and there is a lack of general theory on the minimum linear locality of linear codes. In addition, the minimum linear locality of many known families of linear codes are not studied in the literature. Motivated by these two facts, this paper develops some general theory about the minimum linear locality of linear codes, and investigates the minimum linear
locality of a number of families of linear codes, such as $q$-ary Hamming codes, $q$-ary Simplex codes, generalized Reed-Muller codes, ovoid codes, maximum arc codes, the extended hyperoval codes, and near MDS codes. Many classes of both distance-optimal and dimension-optimal LRCs are presented in this paper. The minimum linear locality of many families of linear codes are settled with the general theory developed in this paper.

The rest of this paper is organized as follows. Section \ref{sec-prelim} introduces some basics of cyclic and linear codes
and the support designs of linear codes. Section \ref{sec-generaltheory} develops some general theory about the minimum
linear locality of nontrivial linear codes. Section  \ref{sec-localityofsome} investigates the minimum linear locality of several
families of famous linear codes, including the $q$-ary Hamming codes, the $q$-ary Simplex codes, the generalized Reed-Muller codes, the ovoid codes, and the maximum arc codes. Section  \ref{sec-localityofNMDS} studies the minimum linear locality of near MDS codes. Section \ref{sec-final}
summarizes the contributions of this paper and makes some concluding remarks.

\section{Preliminaries}\label{sec-prelim}

To study the minimum linear locality of linear codes, we need to introduce some basics of linear codes and cyclic codes.
Since our methodology is combinatorial and group-theoretic, we have to introduce the automorphism groups of linear codes
and combinatorial $t$-designs. The purpose of this section is to introduce these stuffs very briefly.

\subsection{BCH and cyclic codes}

An $[n,k, d]$ code $\C$ over $\gf(q)$ is said to be {\em cyclic} if for each
$(c_0,c_1, c_2, \cdots, c_{n-1}) \in \C$ we have $(c_{n-1}, c_0, c_1, c_2,\cdots, c_{n-2})
\in \C$.
We identify a vector $(c_0,c_1,c_2, \cdots, c_{n-1}) \in \gf(q)^n$
with the polynomial $c(x)=\sum_{i=0}^{n-1} c_ix^i \in \gf(q)[x]/(x^n-1).$
Then a code $\C$ of length $n$ over $\gf(q)$ corresponds to a subset $\C(x)$ of the quotient ring
$\gf(q)[x]/(x^n-1)$, where
$$
\C(x):=\left\{\sum_{i=0}^{n-1} c_ix^i : c=(c_0,c_1, \cdots, c_{n-1}) \in \C\right\}.
$$
It is easy to see that $\C$ is cyclic if and only if the set $\C(x)$ is an ideal of the ring $\gf(q)[x]/(x^n-1)$.

It is well-known that each ideal of $\gf(q)[x]/(x^n-1)$ is principal. Let $\C=\langle g(x) \rangle$ be a
cyclic code, where $g(x)$ is monic and has the smallest degree among all the
generators of $\C$. Then $g(x)$ is unique and called the {\em generator polynomial,}
and $h(x)=(x^n-1)/g(x)$ is referred to as the {\em check polynomial} of $\C$.

Let $n$ be a positive integer with $\gcd(n, q)=1$, and let $m=\ord_{n}(q)$ be the order of $q$ modulo $n$.
Let $\alpha$ be a generator of the multiplicative group $\gf(q^m)^*$. Put $\beta=\alpha^{(q^m-1)/n}$.
Then $\beta$ is a primitive $n$-th root of unity in $\gf(q^m)$. The minimal polynomial $\M_{\beta^s}(x)$
of $\beta^s$ over $\gf(q)$ is defined to be the monic polynomial of the smallest degree over $\gf(q)$
with $\beta^s$ as a root and is given by
\begin{eqnarray}
\M_{\beta^s}(x)=\prod_{i \in C_s} (x-\beta^i) \in \gf(q)[x],
\end{eqnarray}
where $C_s=\{sq^i \bmod n: 0 \leq i \leq m-1\}$ and is called the $q$-cyclotomic class containing $s$.

Let $\delta$ be an integer with $2 \leq \delta \leq n$ and let $h$ be an integer.
A \emph{BCH code\index{BCH codes}} over $\gf(q)$
with length $n$ and \emph{designed distance} $\delta$, denoted by $\C_{(q,n,\delta,h)}$, is a cyclic code with
generator polynomial
\begin{eqnarray}\label{eqn-BCHdefiningSet}
g_{(q,n,\delta,h)}=\lcm(\M_{\beta^h}(x), \M_{\beta^{h+1}}(x), \cdots, \M_{\beta^{h+\delta-2}}(x)).
\end{eqnarray}
If $h=1$, the code $\C_{(q,n,\delta,h)}$ with the generator polynomial in (\ref{eqn-BCHdefiningSet}) is referred to as
a \emph{narrow-sense\index{narrow sense}} BCH code. If $n=q^m-1$, then $\C_{(q,n,\delta,h)}$ is called a \emph{primitive\index{primitive BCH}} BCH code.

BCH codes form a subfamily of cyclic codes with very attractive properties and applications. In many cases BCH codes are the
best linear codes. For instance, among all binary cyclic codes of odd lengths $n$ with $n \leq 125$ the best cyclic code is always a BCH code
except for two special cases \cite{Dingbook15}. Reed-Solomon codes are also BCH codes and have been widely used in data storage systems, communication devices and consumer electronics.

\subsection{Several basic operations on linear codes}

Let $\C$ be a linear code with length $n$. Below we introduce several basic operations on $\C$ for obtaining new codes.
Let $T$ be a set of coordinate positions in $\C$ and let $\C^T$ denote the code obtained by
puncturing $\C$ in all the coordinate positions in $T$, which has length $n-|T|$.
Let $\C(T)$ denote the set of codewords whose coordinates are $\bzero$ on $T$, which
is a subcode of $\C$. After puncturing $\C(T)$ on $T$, we get a linear
code over $\gf(q)$ with length $n-|T|$, which is called a \emph{shortened code}\index{shortened code} of $\C$, and is denoted by $\C_T$.  It is known that $(\C^\perp))_T=(\C^T)^\perp$ and $(\C^\perp)^T=(\C_T)^\perp$.
The \emph{extended code\index{entended code}}
$\overline{\C}$ of $\C$ is defined by
$$
\overline{\C}=\left\{(c_0,c_1,\ldots, c_{n-1}, c_{n}) \in \gf(q)^{n+1}: (c_0,c_1,\ldots, c_{n-1}) \in \C \mbox{ with }
\sum_{i=0}^{n} c_i =0\right\}.
$$
Let $G$ be a generator matrix of $\C$. Suppose that
the all-$1$ vector is not a codeword of $\C$. Then the \emph{augmented code}\index{augmented code}, denoted
by $\widetilde{\C}$, of $\C$ is the linear code over $\gf(q)$ with generator matrix
\begin{eqnarray*}
\left[
\begin{array}{c}
G \\
\bone
\end{array}
\right],
\end{eqnarray*}
where $\bone$ denotes the all-$1$ vector. The augmented code has length $n$ and dimension
$k+1$.

\subsection{Automorphism groups of linear codes}

The \emph{permutation automorphism group\index{permutation automorphism group of codes}} of $\C$,
denoted by $\PAut(\C)$, is the set of coordinate permutations that map a code $\C$ to itself.
A square matrix having exactly one nonzero element of $\gf(q)$  in each row and column is
called a \emph{monomial matrix\index{monomial matrix}} over $\gf(q)$. A monomial matrix $M$ can be written  in
the form $DP$ or the form $PD_1$, where $P$ is a permutation matrix and $D$ and $D_1$ are diagonal matrices.
The \emph{monomial automorphism group\index{monomial automorphism group}} of $\C$ refers to the
set of monomial matrices that map $\C$ to itself. Obviously,  $\PAut(\C) \subseteq \MAut(\C)$.
The \textit{automorphism group}\index{automorphism group} of $\C$, denoted by $\GAut(\C)$, is the set
of maps of the form $M\gamma$ that map $\C$ to itself,
where $M$ is a monomial matrix and $\gamma$ is a field automorphism. If $q=2$,
$\PAut(\C)$,  $\MAut(\C)$ and $\GAut(\C)$ are the same. If $q$ is a prime, $\MAut(\C)$ and
$\GAut(\C)$ are identical. In general, we have
$$
\PAut(\C) \subseteq \MAut(\C) \subseteq \GAut(\C).
$$

By the definitions above, each element in $\GAut(\C)$ is of the form $DP\gamma$, where $D$ is a diagonal matrix,
$P$ is a permutation matrix, and $\gamma$ is an automorphism of $\gf(q)$.
The automorphism group $\GAut(\C)$ is said to be \emph{$t$-transitive}\index{$t$-transitive} if for every pair of $t$-element ordered
sets of coordinates, there is an element $DP\gamma$ of the automorphism group $\GAut(\C)$ such that its
permutation part $P$ sends the first set to the second set.
The automorphism group $\GAut(\C)$ is said to be \emph{$t$-homogeneous} if for every pair of $t$-element sets of coordinates, there is an element $DP\gamma$ of the automorphism group $\GAut(\C)$ such that its
permutation part $P$ sends the first set to the second set. If the automorphism group $\GAut(\C)$ is $t$-transitive,
then it must be $t$-homogeneous. But the converse may not be true. For simplicity, we say that  $\GAut(\C)$ is transitive
(respectively, homogeneous) if $\GAut(\C)$ is $1$-transitive
(respectively, $1$-homogeneous).

\subsection{The support designs of linear codes}

Let $\cP$ be a set of $n$ elements, and let $\cB$ be a set of $k$-subsets of $\cP$, where $1 \leq k \leq n$.
Let $t$ be an integer with $1 \leq t \leq k$. The pair $\bD := (\cP, \cB)$ is an
incidence structure, where the incidence relation is the set membership. The incidence structure
$\bD = (\cP, \cB)$ is called a $t$-$(n, k, \lambda)$ {\em design\index{design}}, or simply {\em $t$-design\index{$t$-design}}, if each $t$-subset of $\cP$ is contained in $\lambda$ elements of
$\cB$. The elements of $\cP$ are referred to as points, and those of $\cB$ are called blocks.
If $\cB$ does not contain
any repeated blocks, then the $t$-design is called {\em simple.\index{simple}}
This paper considers only simple $t$-designs.
A $t$-$(n,k,\lambda)$ design is referred to as a
{\em Steiner system\index{Steiner system}} if $t \geq 2$ and $\lambda=1$,
and is denoted by $S(t,k, n)$.

There are different ways to construct $t$-designs.
A coding-theoretic construction of $t$-designs is briefly described below.
Let $\C$ be a linear code over $\gf(q)$ with length $n$.
For each $k$ with $A_k \neq 0$,  let $\cB_k(\C)$ denote
the set of the supports of all codewords with Hamming weight $k$ in $\C$, where the coordinates of a codeword
are indexed by $(0, 1, \ldots, n-1)$. Let $\cP(\C)=[n]$.  The incidence structure $(\cP(\C), \cB_k(\C))$
may be a $t$-$(n, k, \lambda)$ design for some positive integer $\lambda$, which is called a
\emph{support design} of the code $\C$, and is denoted by $\bD_k(\C)$. In such a case, we say that the codewords of weight $k$
in $\C$ support or hold a $t$-$(n, k, \lambda)$ design, and for simplicity, we say that $\C$ supports or holds a $t$-$(n, k, \lambda)$ design.

The following theorem, called the Assmus-Mattson Theorem,
demonstrates that the pair $(\cP(\C), \cB_k(\C))$ defined by
a linear code $\C$ is a $t$-design under certain conditions \cite{AM69}.

\begin{theorem}\label{thm-designAMtheorem}
Let $\C$ be an $[n,k,d]$ code over $\gf(q)$. Let $d^\perp$ denote the minimum distance of $\C^\perp$.
Let $w$ be the largest integer satisfying $w \leq n$ and
$$
w-\left\lfloor  \frac{w+q-2}{q-1} \right\rfloor <d.
$$
Define $w^\perp$ analogously using $d^\perp$. Let $(A_i)_{i=0}^n$ and $(A_i^\perp)_{i=0}^n$ denote
the weight distribution of $\C$ and $\C^\perp$, respectively. Fix a positive integer $t$ with $t<d$, and
let $s$ be the number of $i$ with $A_i^\perp \neq 0$ for $1 \leq i \leq n-t$. Suppose $s \leq d-t$. Then
\begin{itemize}
\item all the codewords of weight $i$ in $\C$ support a $t$-design provided $A_i \neq 0$ and $d \leq i \leq w$, and
\item all the codewords of weight $i$ in $\C^\perp$ support a $t$-design provided $A_i^\perp \neq 0$ and
         $d^\perp \leq i \leq \min\{n-t, w^\perp\}$.
\end{itemize}
\end{theorem}

The Assmus-Mattson Theorem above is a useful tool in constructing $t$-designs from linear codes
(see, for example, \cite{Dingbook18}),  but does not
characterize all linear codes supporting $t$-designs. The reader is referred to \cite{TDX19} for a generalized
Assmus-Mattson theorem.

Using the automorphism group of a linear code $\C$, the following theorem gives another sufficient condition
for the code $\C$ to hold $t$-designs \cite[p. 308]{HP03}.

\begin{theorem}\label{thm-designCodeAutm}
Let $\C$ be a linear code of length $n$ over $\gf(q)$ such that $\GAut(\C)$ is $t$-transitive
or $t$-homogeneous. Then the codewords of any weight $i \geq t$ of $\C$ hold a $t$-design.
\end{theorem}

\section{General theory about the minimum linear locality of linear codes}\label{sec-generaltheory}

The objective of this section is to develop some general theory about the linear locality of linear codes over finite fields.
In particular, we will answer Questions \ref{quest-bas1} and  \ref{quest-bas2} raised in Section \ref{sec-MOM}. The zero
code $\{\bzero \}$ and the code $\gf(q)^n$ are not interesting in both theory and practice, and are called
\emph{trivial codes}. If the dual distance $d(\C^\perp)$ of $\C$ is $1$, then $\C$ has a zero coordinate, which can be
punctured away without affecting the error-correcting capability. Hence, in theory and practice codes with $d(\C) \leq
1$ or $d(\C^\perp) \leq 1$ are not interesting, and are called trivial codes. In this paper, we consider the minimum linear locality
only for nontrivial linear codes, i.e., linear codes with  $d(\C) \geq 2$ and $d(\C^\perp) \geq 2$.  Recall that
we use the elements
in $[n]$ to index the coordinate positions in a linear code of length $n$.

\subsection{Some general theory the linear locality of nontrivial linear codes}

The following lemma directly follows from the definition of linear locality of linear codes, and is well known in the literature.
We will need it later.

\begin{lemma}\label{lem-localitylem}
Let $\C$ be a nontrivial linear code of length $n$. Then $\C$ has linear locality $r$ if and only if for each $i \in [n]$ the dual code $\C^\perp$ has a codeword $\bc$ of Hamming weight at most $r+1$ such that $i \in \support(\bc)$.
\end{lemma}

\begin{theorem}\label{thm-generalLoc}
Let $\C$ be a nontrivial linear code of lenth $n$. Then there exists a positive integer $w$ with $2 \leq w \leq n$ such that $A_w(\C^\perp)>0$ and
\begin{eqnarray}\label{eqn-condloc}
\bigcup_{j=1}^w \bigcup_{S \in \cB_j(\C^\perp)} S \supseteq [n].
\end{eqnarray}

Let $w$ be the smallest integer such that $A_w(\C^\perp)>0$ and (\ref{eqn-condloc}) holds. Then
$\C$ has minimum linear locality $w-1$.
\end{theorem}

\begin{proof}
Suppose that there is an integer $i$ in $[n]$ such that
\begin{eqnarray*}
i \not\in \bigcup_{j=1}^n \bigcup_{S \in \cB_j(\C^\perp)} S.
\end{eqnarray*}
Then $c_i^\perp=0$ for all codewords $\bc^\perp=(c_0^\perp, \cdots, c_{n-1}^\perp)$ in $\C^\perp$.
Consequently, the vector $\bc=(0,\cdots, 0, 1, 0, \cdots, 0)$ of length $n$, which has only one nonzero coordinate $1$ in coordinate position $i$,
is a codeword in $\C$. This is contrary to the assumption that $d \geq 2$.

Let $w$ be the smallest integer such that $A_w(\C^\perp)>0$ and (\ref{eqn-condloc}) holds.
Then every integer $i \in [n]$ is contained in $\support(\bc^\perp)$, where $\bc^\perp$
is some codeword with weight at most $w$ in $\C^\perp$.  Then the code symbol $c_i$ in $\C$
can be recovered by a linear combination of the coordinates in the positions in $\support(\bc^\perp)\setminus \{i\}$.
Then $\C$ has linear locality $w-1$.

If the code symbol $c_i$ in every codeword $\bc$ in $\C$ can be recovered linearly by
$$
c_i=u_1c_{i_1}+ \cdots + u_hc_{i_h}
$$
where $u_i \neq 0$ and $R_i=\{i_1, \ldots, i_h\}$ is the corresponding recovering set of the code symbol $c_i$.
Then $\C^\perp$ has a codeword with weight $h+1$. Hence, $w-1$ is the minimum
linear locality.
\end{proof}

Theorem \ref{thm-generalLoc} means that every nontrivial linear code has a minimum linear locality,
and tells us how to calculate the minimum linear locality. In practice, it is also necessary and important
to find a recovering set $R_i$ for each code symbol $c_i$. But we will not deal with this problem in
this paper.

\subsection{Linear codes $\C$ with minimum linear locality $d(\C^\perp)-1$}

It follows from Lemma \ref{lem-localitylem} that the minimum linear locality of a nontrivial linear code $\C$ is at least
$d(\C^\perp)-1$. Hence, nontrivial linear codes $\C$ with minimum linear locality $d(\C^\perp)-1$ would be very
interesting in both theory and practice. In this subsection, we develop some general results for such special
codes. It will be seen later that there are indeed nontrivial linear codes $\C$  with minimum linear locality
more than $d(\C^\perp)-1$.

\begin{corollary}\label{cor-newn0}
Let $\C$ be a nontrivial linear code of length $n$ and put $d^\perp=d(\C^\perp)$.
Then $\C$ has minimum linear locality $d^\perp-1$ if and only if
\begin{eqnarray}\label{eqn-condj17}
\bigcup_{S \in \cB_{d^\perp}(\C^\perp)} S =[n].
\end{eqnarray}
\end{corollary}

\begin{proof}
The desired conclusion directly follows from Theorem \ref{thm-generalLoc} and Lemma \ref{lem-localitylem}.
\end{proof}

The following result is well known in the literature. We show that it is a corollary of Theorem \ref{thm-generalLoc}.

\begin{corollary}\label{cor-newnc}
Let $\C$ be a nontrivial cyclic code of length $n$. Then $\C$ has minimum linear locality $d(\C^\perp)-1$.
\end{corollary}

\begin{proof}
Put $d^\perp=d(\C^\perp)$.
Let $i \in [n]$.
Let $\bc^\perp$ be a minimum weight codeword in $\C^\perp$. By definition, the Hamming weight $\wt(\bc^\perp) \geq 2$.
Consequently, $\bc^\perp$ has a nonzero coordinate.  Since $\C^\perp$ is also cyclic,
we can assume $\bc^\perp_i \neq 0$. We then deduce that
$$
\bigcup_{S \in \cB_{d^\perp}(\C^\perp)} S \supseteq\{i\}.
$$
The desired conclusion then follows from Theorem \ref{thm-generalLoc}.
\end{proof}

\begin{corollary}\label{cor-newn}
Let $\C$ be a nontrivial linear code of length $n$ and put $d^\perp=d(\C^\perp)$.
If $(\p(\C^\perp), \B_{d^\perp}(\C^\perp))$ is a $1$-$(n, d^\perp, \lambda_1^\perp)$ design with $\lambda_1^\perp \geq 1$,
then $\C$ has minimum linear locality $d^\perp-1$.
\end{corollary}

\begin{proof}
By the definition of $1$-designs, every $i \in \cP(\C^\perp)$ is covered in $\lambda_1^\perp$ blocks in
the block set $\cB_{d^\perp}(\C^\perp)$. Hence,
\begin{eqnarray*}
\bigcup_{j=1}^{d^\perp} \bigcup_{S \in \cB_j(\C^\perp)} S =
\bigcup_{S \in \cB_{d^\perp}(\C^\perp)} S = [n].
\end{eqnarray*}
The desired conclusion then follows from Theorem \ref{thm-generalLoc}.
\end{proof}

It should be noted that there are many nontrivial linear codes with minimum linear locality  $d(\C^\perp)-1$, but
$(\p(\C^\perp), \B_{d^\perp}(\C^\perp))$ is not a $1$-design. Hence, the converse of Corollary \ref{cor-newn} is
not true. Corollary  \ref{cor-newn} will be one of the tools for studying the minimum linear locality of some
families of linear codes in this paper. Another tool is documented in the following corollary.

\begin{corollary}\label{cor-newtrans}
Let $\C$ be a nontrivial linear code. If $\Aut(\C)$ or $\Aut(\C^\perp)$ is transitive, then $\C$ has minimum linear locality
$d(\C^\perp)-1$ and $\C^\perp$ has minimum linear locality $d(\C)-1$.
\end{corollary}

\begin{proof}
Put $d^\perp=d(\C^\perp)$.
Let $\C$ be over $\gf(q)$ and have length $n$. Suppose that  $\Aut(\C^\perp)$ is transitive. Let $\bc^\perp$ be a minimum
weight codeword in $\C^\perp$. Then $\wt(\bc^\perp) \geq 2$. Let $i \in \support(\bc^\perp)$. For each $j \in [n] \setminus
\{i\}$, there is an automorphism $DP\gamma$ in $\Aut(\C^\perp)$ such that the permutation part $P$ sends $i$ to $j$,
as $\Aut(\C^\perp)$ is transitive. This means there is another minimum weight codeword $(\bc')^\perp$ in $\C^\perp$
such that $j \in \support((\bc')^\perp)$. Consequently,
$$
\bigcup_{j=1}^{d^\perp} \bigcup_{S \in \cB_j(\C^\perp)} S =
\bigcup_{S \in \cB_{d^\perp}(\C^\perp)} S = [n].
$$
It then follows from Theorem \ref{thm-generalLoc} that $\C$ has minimum linear locality $d(\C^\perp)-1$.

In general,  $\Aut(\C)$ and $\Aut(\C^\perp)$ are different. However, it is straightforward to prove that
$\Aut(\C)$ is transitive if and only if $\Aut(\C^\perp)$ is so. Then the  remaining desired conclusion follows
from the first conclusion proved above.
\end{proof}

Note that combining Theorem \ref{thm-designCodeAutm} and Corollary \ref{cor-newn} gives another proof of
Corollary \ref{cor-newtrans}.
Sometimes we may need to use Corollary \ref{cor-newn}, as the automorphism group of a code may be unknown.
Sometimes it is more convenient to use Corollary  \ref{cor-newtrans}. Sometimes both corollaries can be used to
study the linear locality of some linear codes. In many cases, both corollaries cannot be used to do so.
It looks impossible to find out all nontrivial linear codes with minimum linear locality
$d(\C^\perp)-1$. But  Corollaries \ref{cor-newn} and  \ref{cor-newtrans} can be employed to find many families
of such codes. Most of the families of linear codes documented in the monograph \cite{Dingbook18} are such codes,
as they support $t$-designs with $t \geq 2$ or their automorphism groups are doubly homogeneous.
Other families of such linear codes are not documented in \cite{Dingbook18}, as the monograph \cite{Dingbook18}
does not include linear codes supporting $1$-designs but not $2$-designs.

The following result would also be useful in some cases.

\begin{theorem}\label{thm-newJ1821}
Let $\C$ be a nontrivial linear code. If $\C^\perp$ is spanned by its minimum weight codewords, then
$\C$ has minimum linear locality $d(\C^\perp)-1$.
\end{theorem}

\begin{proof}
Let $\C$ have length $n$. Let $i \in [n]$.  If
$$
i \not\in \bigcup_{S \in \cB_{d^\perp}(\C^\perp)} S,
$$
then $(0,\ldots, 0, 1, 0, \ldots, 0)$ would be a codeword in $\C$, where the nonzero coordinate $1$ is in coordinate
position $i$, as all the minimum weight codewords in $\C^\perp$ span $\C^\perp$. This is contrary to the fact that
$\C$ is nontrivial. The desired conclusion then follows from Corollary \ref{cor-newn0}.
\end{proof}

\subsection{The minimum linear locality of extended cyclic code}

While any nontrivial cyclic code $\C$ has minimum linear locality $d(\C^\perp)-1$,  extended
cyclic codes may not have such property. Note that even if $\C$ is nontrivial, the extended code
$\overline{\C}$ may be trivial, as $d( \overline{\C})$ could be $1$.
The automorphism group of any cyclic code is transitive and each cyclic code supports $1$-designs.
But these may not be true for extended cyclic codes.
In this section, we consider
the linear locality of the extended cyclic codes and their duals.

Let $H$ and $\overline{H}$ denote the parity-check matrix of $\C$ and $\overline{\C}$, respectively.
Then we have the following well known lemma \cite{HP03}.

\begin{lemma}\label{thm-extendedCodeParam}
Let $\C$ be an $[n, \kappa, d]$ code over $\gf(q)$. Then $\overline{\C}$ is an $[n+1, \kappa, \overline{d}]$ linear code, where $\overline{d}=d$ or $d+1$. In the binary case, $\overline{d}=d$ if $d$ is even, and
$\overline{d}=d+1$ otherwise.

In addition, the parity-check matrix $\overline{H}$ of $\overline{\C}$ can be deduced from that of $\C$
by
\begin{eqnarray}
\overline{H}=\left[
\begin{array}{ll}
\bone & 1 \\
H &     \bzero
\end{array}
\right],
\end{eqnarray}
where $\bone=(1,1, \ldots, 1)$ and $\bzero=(0,0, \ldots, 0)^T$.
\end{lemma}

We now prove the following result, which will be needed later.

\begin{theorem}\label{thm-loca-ext-cycliccode}
Let $\C$ be a nontrivial cyclic code. If $d(\overline{\C})=d(\C)+1$, then $(\overline{\C})^\perp$ has minimum linear
locality $d(\C)$.
\end{theorem}

\begin{proof}
Let $\C$ have length $n$. By definition, $d(\C) \geq 2$ and $d(\C^\perp) \geq 2$. Since $d(\overline{\C})=d(\C)+1$,
we know that $d((\overline{\C})^\perp) \geq 2$. Therefore, $\overline{\C}$ is nontrivial. Let $\bc_1, \ldots, \bc_h$ be
all the minimum weight codewords in $\C$, and let $\overline{\bc_i}$ be the extended codeword of $\bc_i$ in $\overline{\C}$.
Since $d(\C)\geq 2$ and $\C$ is cyclic, we have
$$
\bigcup_{i=1}^h \support(\bc_i) =[n].
$$
Since $d(\overline{\C})=d(\C)+1$, the extended coordinate in each  $\overline{\bc_i}$ is nonzero. As a result, we get
$$
\bigcup_{i=1}^h \support(\overline{\bc_i}) =[n+1].
$$
The desired conclusion then follows from Corollary \ref{cor-newn0}.
\end{proof}

\begin{corollary}\label{cor-localityextedbinarycode}
Let $\C$ be a nontrivial binary cyclic code. If $d(\C)$ is odd, then $(\overline{\C})^\perp$ has minimum linear
locality $d(\C)$.
\end{corollary}

\begin{proof}
The desired conclusion follows from Lemma \ref{thm-extendedCodeParam} and Theorem \ref{thm-loca-ext-cycliccode}.
\end{proof}

Corollary \ref{cor-localityextedbinarycode} has determined the minimum linear locality of  $(\overline{\C})^\perp$
for all binary cyclic codes. Specifically, either $(\overline{\C})^\perp$ is a trivial code or nontrivial binary linear code
with minimum linear locality $d(\C)$ for each nontrivial cyclic code $\C$.

\section{The minimum linear locality of some known families of linear codes}\label{sec-localityofsome}

The objective of this section is to study the minimum linear locality of several families of linear codes which
are geometric codes and their punctured and shortened codes.
We wish to find out some families of optimal LLRCs.

\subsection{The minimum linear locality of the $q$-ary Hamming codes and Simplex codes}

A parity check matrix $H_{(q,m)}$ of the Hamming code $\cH_{(q,m)}$ over $\gf(q)$ is defined by choosing
for its columns a nonzero vector from each one-dimensional subspace of $\gf(q)^m$. In terms of
finite geometry, the columns of $H_{(q,m)}$ are the points of the projective geometry $\PG(m-1,\gf(q))$.
Hence $\cH_{(q,m)}$ has length $n=(q^m-1)/(q-1)$ and dimension $n-m$.
It is well known that $\cH_{(q,m)}$ has minimum weight 3 and
any $[(q^m-1)/(q-1), (q^m-1)/(q-1)-m,3]$ code over $\gf(q)$ is monomially equivalent
to the Hamming code $\cH_{(q,m)}$ \cite{HP03}. Note that the Hamming code $\cH_{(q,m)}$ is permutation-equivalent
to a cyclic code when $\gcd(m,q-1)=1$. By Corollary \ref{cor-newnc}, its minimum linear locality is known in this case.
However, its minimum locality may not be known for the case that $\gcd(m,q-1)\neq1$. The linear locality of the binary
Hamming and Simplex codes was settled in \cite{HYUS20}. Note that binary Hamming and Simplex codes are
permutation-equivalent to cyclic codes. In this subsection, we investigate the minimum linear locality of the $q$-ary
Hamming and Simplex codes.

The weight distribution of $\cH_{(q,m)}$ is given in the following lemma \cite{LDT20}.

\begin{lemma}\label{lem-hamweight}
The weight distribution of $\cH_{(q,m)}$ is given by
\begin{eqnarray*}
  q^mA_k(\cH_{(q,m)}) &=& \sum_{\substack{0\leq i\leq\frac{q^{m-1}-1}{q-1}\\ 0\leq j\leq q^{m-1}\\ i+j=k }} \left[{\frac{q^{m-1}-1}{q-1}\choose i}{q^{m-1}\choose j}\bigg((q-1)^k+(-1)^j(q-1)^i(q^m-1)\bigg)\right]
\end{eqnarray*}
for $0\leq k\leq (q^m-1)/(q-1)$.
\end{lemma}

The duals of the Hamming codes $\cH_{(q,m)}$ are called Simplex codes, denoted by $\Sim_{(q,m)}$, which have parameters
$[(q^m-1)/(q-1),m,q^{m-1}]$. The nonzero codewords of the $[(q^m-1)/(q-1),m,q^{m-1}]$ Simplex
codes all have weight $q^{m-1}$.

\begin{theorem}\label{thm-hamloc}
The Hamming code $\cH_{(q,m)}$ is an {$(n,n-m,3,q;q^{m-1}-1)$}-LLRC and the Simplex code $\Sim_{(q,m)}$ is an
$(n,m,q^{m-1},q;2)$-LLRC. Furthermore, the Hamming code $\cH_{(q,m)}$ and $\Sim_{(q,m)}$ are $k$-optimal.
\end{theorem}

\begin{proof}
The Hamming code $\cH_{(q,m)}$ has parameters $[n, n-m,3]$ and its dual code is a one-weight code. Then by the Assmus-Mattson Theorem, the codewords of minimum weight in the  Hamming code and Simplex code both hold a 2-design. So $(\p(\cH_{(q,m)}), \B_{3}(\cH_{(q,m)}))$ and $(\p(\cH_{(q,m)}^\perp), \B_{q^{m-1}}(\cH_{(q,m)}^\perp))$ are $1$-designs. Hence, the conclusions on the minimum linear locality of the two codes follow from Corollary \ref{cor-newn}.

We now prove the dimension optimality of $\cH_{(q,m)}$. Putting $t=1$ and the parameters of the {$(n,n-m,3,q; q^{m-1}-1)$}-LLRC into
the right-hand side of the CM bound in \eqref{cm-bound}, we have
\begin{align*}
  k & \leq \min_{s \in\Z_+}\{rs+k_{opt}^{(q)}(n-(r+1)s,d)\}\\
          & \leq r+k_{opt}^{(q)}(n-(r+1),d) \\
          & =q^{m-1}-1+k_{opt}^{(q)}(n-q^{m-1},3)\\
          & \leq n-m,
\end{align*}
where the last inequality holds due to the fact that $k_{opt}^{(q)}(n-q^{m-1},3)\leq n-q^{m-1}-m+1$, which follows from the sphere packing bound. Therefore, the Hamming code $\cH_{(q,m)}$ is $k$-optimal.

Putting $t=1$ and the parameters of $(n,m,q^{m-1},q;2)$-LLRC into the right-hand side of the CM bound in \eqref{cm-bound}, we have
\begin{align*}
  k & \leq r+k_{opt}^{(q)}(n-(r+1),d) \\
          & =2+k_{opt}^{(q)}(n-3,q^{m-1})\\
          & \leq m,
\end{align*}
where the last inequality holds due to the fact that $k_{opt}^{(q)}(n-3,q^{m-1})\leq m-2$, which follows from the Plotkin bound. Therefore, the Simplex code $\Sim_{(q,m)}$ is $k$-optimal.
This completes the proof.
\end{proof}

According to the Singleton-like bound in \eqref{eq-bound}, we can obtain the following family of $d$-optimal LLRCs.

\begin{corollary}
When $m=3$, the Hamming code $\cH_{(q,3)}$ is a $(q^2+q+1,q^2+q-2,3,q; q^{2}-1)$-LLRC and is both $d$-optimal and $k$-optimal.
\end{corollary}

\begin{proof}
The parameters and dimension optimality of $\cH_{(q,3)}$ directly follow from Theorem \ref{thm-hamloc}. Hence, we only need to prove the distance optimality. It is easy to verify that the parameters of $\cH_{(q,3)}$ satisfy the equality in \eqref{eq-bound}. Hence, it is a $d$-optimal $(q^2+q+1,q^2+q-2,3,q;q^{2}-1)$-LLRC. This completes the proof.
\end{proof}

If a linear code $\C$ supports 2-designs, then the punctured code $\C^{\{t_1\}}$ or shortened code $\C_{\{t_1\}}$ may
support 1-designs. Then we can settle the minimum linear locality of $\C^{\{t_1\}}$ or $\C_{\{t_1\}}$.
The parameters of some punctured codes and shortened codes of the Hamming code are given in the following
lemma \cite{LDT20}.

\begin{lemma}\label{lem-shorham}
Let $n=(q^m-1)/(q-1)\geq 4$, and let $t_1$ be any coordinate position of codewords in $\cH_{(q,m)}$. Then the following hold:
\begin{itemize}
  \item $(\cH_{(q,m)})_{\{t_1\}}$ is an $[n-1,n-m-1,3]$ code over $\gf(q)$ with
      \begin{eqnarray*}
        A_k((\cH_{(q,m)})_{\{t_1\}}) &=& \frac{n-k}{n}A_k(\cH_{(q,m)})
      \end{eqnarray*}
      for $0\leq k\leq n-1$, where $A_k(\cH_{(q,m)})$ was given in Lemma \ref{lem-hamweight}.
  \item $((\cH_{(q,m)})_{\{t_1\}})^\perp$ is an $[n-1,m,q^{m-1}-1]$ code over $\gf(q)$ with weight enumerator $$1+(q-1 )q^{m-1}z^{q^{m-1}-1}+(q^{m-1}-1)z^{q^{m-1}}.$$
  \item $(\Sim_{(q,m)})_{\{t_1\}}$ is an $[n-1,m-1,q^{m-1}]$ code over $\gf(q)$ with weight enumerator $1+(q^{m-1}-1)z^{q^{m-1}}.$
  \item $((\Sim_{(q,m)})_{\{t_1\}})^\perp$ is an $[n-1,n-m,2]$ code over $\gf(q)$ with weight enumerator $$\frac{1}{q^{m-1}}[(1+(q-1)z)^{n-1}+(q^{m-1}-1)(1-z)^{q^{m-1}} (1+(q-1)z)^{n-1-q^{m-1}}].$$
\end{itemize}
\end{lemma}

With the Assmus-Mattson Theorem, we can deduce that the codewords of minimum weight in these codes in Lemma \ref{lem-shorham} hold a 1-design. Then by Corollary \ref{cor-newn}, we can settle the minimum linear locality of these codes in Lemma \ref{lem-shorham}.

\begin{theorem}\label{thm-hmcodeps}
Let $n=(q^m-1)/(q-1)\geq 4$, and let $t_1$ be any coordinate position of codewords in $\cH_{(q,m)}$. Then
we have the following.
\begin{itemize}
  \item $(\cH_{(q,m)})_{\{t_1\}}$ is a $k$-optimal $(n-1,n-m-1,3,q;q^{m-1}-2)$-LLRC.
  \item $((\cH_{(q,m)})_{\{t_1\}})^\perp$ is a $k$-optimal $(n-1,m,q^{m-1}-1,q;2)$-LLRC.
  \item $(\Sim_{(q,m)})_{\{t_1\}}$ is a $k$-optimal  $(n-1,m-1,q^{m-1},q;1)$-LLRC.
  \item $(\Sim_{(q,m)})_{\{t_1\}})^\perp$ is a $k$-optimal  $(n-1,n-m,2,q;q^{m-1}-1)$-LLRC.
\end{itemize}
\end{theorem}

\begin{proof}
The conclusions on the parameters of the codes follow from Lemma \ref{lem-shorham}, the Assmus-Mattson Theorem and Corollary \ref{cor-newn}. The proofs of the dimension optimality of $(\cH_{(q,m)})_{\{t_1\}}$ and $(\Sim_{(q,m)})_{\{t_1\}})^\perp$ are similar to that in Theorem \ref{thm-hamloc}, and are
omitted.
\end{proof}

Furthermore, we can obtain the following $d$-optimal LLRCs.

\begin{corollary}\label{cor-21Janu213}
The code $(\cH_{(q,3)})_{\{t_1\}}$ is a $(q^2+q,q^2+q-3,3,q; q^2-2)$-LLRC and the code
$((\Sim_{(q,3)})_{\{t_1\}})^\perp$ is a $(q^2+q,q^2+q-2,2, q; q^2-1)$-LLRC. Furthermore, they are both $d$-optimal and $k$-optimal.
\end{corollary}

\begin{proof}
The parameters of the codes follow from Theorem \ref{thm-hmcodeps}. The $d$-optimality and $k$-opmality of the codes
are with respect to the Singleton-like bound and CM bound and can be easily verified.
\end{proof}

\subsection{The minimum linear locality of the generalised Reed-Muller codes over $\gf(q)$}

The general affine group $\GA_1(\gf(q^m))$ is defined by
$$
\GA_1(\gf(q^m))=\{ax+b: a \in \gf(q^m)^*, \ b \in \gf(q^m) \},
$$
which acts on $\gf(q^m)$ doubly transitively \cite[Section 1.7]{Dingbook18}.

We can index the coordinates of a linear code of length $q^m$ with the elements of $\gf(q^m)$. When each permutation
in $\GA_1(\gf(q^m))$ is applied to a codeword, it is applied to the indices of the coordinates.  A linear code $\C$
of length $q^m$ is said to be affine-invariant if $\GA_1(\gf(q))$ fixes $\C$. It follows from Theorem
\ref{thm-designCodeAutm} that affine-invariant codes supports $2$-designs. By Corollary \ref{cor-newtrans},
all affine-invariant codes $\C$ have minimum locality $d(\C^\perp)-1$. There are many infinite families of affine-invariant
codes \cite[Chapter 6]{Dingbook18}. Our objective in this section is to study the minimum linear locality of the generalised Reed-Muller
codes and obtain a class of either $k$-optimal or almost $k$-optimal LLRCs.

For any integer $j=\sum_{i=0}^{m-1}j_iq^i$, where $0 \leq j_i
\leq q-1$ for all $0 \leq i \leq m-1$ and $m$ is a positive integer, we define
\begin{eqnarray}\label{eqn-qweight}
\wt_q(j)=\sum_{i=0}^{m-1} j_i,
\end{eqnarray}
where the sum is taken over the ring of integers, and is called the $q$-weight of $j$.
Let $\ell$ be a positive integer with $1 \leq \ell <(q-1)m$. The $\ell$-th order
\emph{punctured generalized Reed-Muller code}\index{punctured generalized Reed-Muller code}
$\cR_q(\ell, m)^*$ over $\gf(q)$ is the cyclic code of length $n=q^m-1$ with generator polynomial
\begin{eqnarray}\label{eqn-generatorpolyPGRMcode}
g(x) = \prod_{\myatop{1 \leq j \leq n-1}{ \wt_q(j) < (q-1)m-\ell}} (x - \alpha^j),
\end{eqnarray}
where $\alpha$ is a generator of $\gf(q^m)^*$. Since $\wt_q(j)$ is a constant function on
each $q$-cyclotomic coset modulo $n=q^m-1$, $g(x)$ is a polynomial over $\gf(q)$.

The generalized Reed-Muller code $\cR_q(\ell, m)$ is defined to be the extended code
of $\cR_q(\ell, m)^*$, and its parameters are given below \cite{AK98}.

\begin{theorem}\label{thm-GRMcodePara}
Let $0 \leq \ell < q(m-1)$. Then the generalized Reed-Muller code $\cR_q(\ell, m)$ has length $n=q^m$, dimension
\begin{eqnarray}\label{eqn-needed181}
\kappa=\sum_{i=0}^\ell \sum_{j=0}^{m} (-1)^j \binom{m}{j} \binom{i-jq+m-1}{i-jq},
\end{eqnarray}
and minimum weight
\begin{eqnarray}\label{eqn-needed182}
d = (q-\ell_0)q^{m-\ell_1-1},
\end{eqnarray}
where $\ell=\ell_1 (q-1)+\ell_0$ and
$0 \leq \ell_0 < q-1$.
\end{theorem}

The minimum linear locality of the generalized Reed-Muller code $\cR_q(\ell, m)$ is given in the following theorem.

\begin{theorem}\label{thm-rmlocalitymain}
Let $0 \leq \ell < q(m-1)$ and $m(q-1)-1-\ell=\ell'_1(q-1)+\ell'_0$ with $0 \leq \ell'_0 < q-1$.
The generalized Reed-Muller code $\cR_q(\ell, m)$ is a $[q^m, \kappa, d,  q; (q-\ell'_0)q^{m-\ell'_1-1}-1 ]$ LLRC, where $\kappa$ and $d$ were
given in (\ref{eqn-needed181}) and  (\ref{eqn-needed182}), respectively.
\end{theorem}

\begin{proof}
The dimension and minimum weight of the code were given in Theorem \ref{thm-GRMcodePara}. We only prove its
minimum linear locality.
It is well known that the generalised Reed-Muller code $\cR_q(\ell, m)$ is affine-invariant \cite[Chapter 6]{Dingbook18}.
By Corollary  \ref{cor-newtrans}, the code has minimum duality $d(\cR_q(\ell, m)^\perp)-1$.
It was proved in \cite{AK98} that
\begin{eqnarray}\label{eqn-mnjan21}
\cR_q(\ell, m)^\perp = \cR_q(m(q-1)-1-\ell, m).
\end{eqnarray}
It then follows from Theorem \ref{thm-GRMcodePara} that  the minimum linear locality is
$$
r=d(\cR_q(\ell, m)^\perp)-1=d(\cR_q(m(q-1)-1-\ell, m))-1=(q-\ell'_0)q^{m-\ell'_1-1}-1.
$$
This completes the proof.
\end{proof}

\begin{corollary}\label{cor-2021J214}
Let $q>2$. Then
$\cR_q(1,m)$ is a $(q^m,1+m, (q-1)q^{m-1}, q; 2)$-LLRC and its dual $\cR_q(1,m)^\perp$ is a
$(q^m,q^m-1-m, 3,q; (q-1)q^{m-1}-1)$-LLRC. Both codes are
$k$-optimal.
\end{corollary}

\begin{proof}
The parameters of the two codes were given in Theorem \ref{thm-rmlocalitymain}.
Putting the parameters of $\R_q(1,m)^\perp$ into the right-hand side of the CM bound in \eqref{cm-bound}, we have
\begin{align*}
  k & \leq \min_{s\in\Z_+}\{rs+k_{opt}^{(q)}(n-(r+1)s,d)\}\\
          & \leq r+k_{opt}^{(q)}(n-(r+1),d) \\
          & = (q-1)q^{m-1}-1+k_{opt}^{(q)}(q^{m-1},3)\\
          & \leq q^m-m-1,
\end{align*}
where the last inequality holds due to the fact that $k_{opt}^{(q)}(q^{m-1},3)\leq q^{m-1}-m+1$, which follows from the Sphere packing bound. Therefore, the $\R_q(1,m)^\perp$ is $k$-optimal.

Taking the parameters of $\cR_q(1,m)$ on the right-hand side of the CM bound in \eqref{cm-bound}, we have
                    \begin{align*}
                          k & \leq \min_{s\in\Z_+}\{rs+k_{opt}^{(q)}(n-(r+1)s,d)\}\\
                             & \leq r+k_{opt}^{(q)}(n-(r+1),d) \\
                             & = 2+k_{opt}^{(q)}(q^m-3,(q-1)q^{m-1})\\
                             & \leq m+1,
                    \end{align*}
                  where the last inequality holds due to the fact that $k_{opt}^{(q)}(q^m-3,(q-1)q^{m-1})\leq m-1$, which follows from the Plotkin bound. Therefore, $\cR_q(1,m)$ is $k$-optimal. This completes the proof.
\end{proof}

In this section, we found two classes of affine-invariant codes $\cR_1(1,m)$ and $\cR_q(1,m)^\perp$, which are $k$-optimal. It would be nice if other classes of $d$-optimal or $k$-optimal affine-invariant codes could be found. Notice that
many classes of affine-invariant codes are known in the literature.

\subsection{The minimum linear locality of ovoid codes}

In the projective space PG$(3,\gf(q))$ with $q >2$, an ovoid $\V$ is a set of $q^2+1$ points such that
no three of them are collinear (i.e., on the same line). In other words, an ovoid is a $(q^2+1)$-cap
(a cap with $q^2+1$ points) in PG$(3,\gf(q))$, and thus a maximal cap. Two ovoids are said to be
equivalent if there is a collineation (i.e., automorphism) of PG$(3,\gf(q))$ that sends one to the
other.

A classical ovoid $\V$ can be defined as the set of all points given by
\[
\V = \{(0,0,1,0)\}\cup\{(x, y, x^2+xy+ay^2, 1) : x, y\in\gf(q)\}
\]
where $a\in\gf(q)$ is such that the polynomial $x^2+x+a$ has no root in $\gf(q)$. Such ovoid is
called an elliptic quadric, as the points come from a non-degenerate elliptic quadratic form.

For $q = 2^{2e+1}$ with $e\geq 1$, there is an ovoid which is not an elliptic quadric, and is called the
Tits ovoid. It is
defined by
\[
\T=\{(0,0,1,0)\}\cup\{(x,y,x^\sigma+xy+y^{\sigma+2},1): x,y\in\gf(q)\},
\]
where $\sigma=2^{e+1}$.

Let $\V$ be an ovoid in PG$(3,\gf(q))$ with $q > 2$. Denote
$
\V = \{\bv_1,\bv_2,\ldots,\bv_{q^2+1}\},
$
where each $\bv_i$ is a column vector in $\gf(q)^4$. Let $\C_\V$ be the linear code over $\gf(q)$ with generator
matrix
$
G_\V=[\bv_1\bv_2\cdots\bv_{q^2+1}].
$
It is known that $\C_\V$ is a $[q^2+1,4,q^2-q]$ code over $\gf(q)$ with weight enumerator
\[
1+(q^2-q)(q^2+1)z^{q^2-q}+(q-1)(q^2+1)z^{q^2}
\]
and its dual $\C_\V^\perp$ is a $[q^2+1,q^2-3,4]$ almost MDS code over $\gf(q)$ \cite[Chapter 13]{Dingbook18}. Conversely, the set of
column vectors of a generator matrix of any $[q^2+1,4,q^2-q]$ code over $\gf(q)$ is an ovoid in
PG$(3,\gf(q))$. Hence, ovoids in PG$(3,\gf(q))$ and $[q^2+1,4,q^2-q]$ codes over $\gf(q)$ are
equivalent in the sense that one can be used to construct the other, and a $[q^2+1,4,q^2-q]$ code over $\gf(q)$ is called an ovoid code over $\gf(q)$.

Ovoid codes are very interesting in combinatorics, as they support $3$-designs, which are documented below
\cite[Chapter 13]{Dingbook18}.

\begin{lemma}\label{lem-22}
The supports of all minimum weight codewords in an ovoid code form a 3-$(q^2+1,q^2-q,(q-2)(q^2-q-1))$ design and the supports of all codewords of weight 4 in the dual of the ovoid code form a 3-$(q^2+1,4,q-2)$ design.
\end{lemma}

The linear locality of an ovoid code and its dual is described in the next theorem.

\begin{theorem}\label{thm-ovoidloca21}
An ovoid code $\C_o$ is a $(q^2+1,4,q;q^2-q,3)$-LLRC and its dual $\C_o^\perp$ is a $(q^2+1,q^2-3,4,q;q^2-q-1)$-LLRC. Moreover, $\C_o$ is $k$-optimal and $\C_o^\perp$ is $d$-optimal and $k$-optimal.
\end{theorem}

\begin{proof}
The parameters follow from Lemma \ref{lem-22} and Corollary \ref{cor-newn}. It is easy to check the distance optimality of $\C_o^\perp$. Then we check the dimension optimality of $\C_o^\perp$. Putting $t=1$ into the right-hand side of the CM bound in \eqref{cm-bound}, one has
\begin{align*}
  k & \leq r+k_{opt}^{(q)}(n-(r+1),d) \\
          & = (q-1)q-1+k_{opt}^{(q)}(q+1,4)\\
          & \leq q^2-3,
\end{align*}
where the last inequality holds due to the fact that $k_{opt}^{(q)}(q+1,4)\leq q-2$, which follows from the classical Singleton bound. Therefore, the code $\C_o^\perp$ is $k$-optimal.

Taking the parameters of $\C_o$ on the right-hand side of the CM bound in \eqref{cm-bound}, one has
\begin{align*}
  k & \leq \min_{t\in\Z_+}\{rt+k_{opt}^{(q)}(n-(r+1)t,d)\}\\
          & \leq r+k_{opt}^{(q)}(n-(r+1),d) \\
          & = 3+k_{opt}^{(q)}(q^2-3,q^2-q)\\
          & \leq 4,
\end{align*}
where the last inequality holds due to the fact that $k_{opt}^{(q)}(q^2-3,q^2-q)\leq 1$, which follows from the Plotkin bound. Therefore, $\C_o$ is $k$-optimal. This completes the proof.
\end{proof}

In \cite{LDT20}, Liu et al. studied some shortened and punctured codes of an ovoid code, and obtained the following results.

\begin{lemma}\label{lem-LDT20}
Let $q\geq 4$, and let $\C_o$ be a $[q^2+1,4,q^2-q]$ code over $\gf(q)$. For any coordinate position $\{t_1\}$, the following hold.
\begin{itemize}
  \item $(\C_o)_{\{t_1\}}$ is a $[q^2,3,q^2-q]$ code over $\gf(q)$ with weight enumerator
      \[
1+q(q^2-1)z^{q^2-q}+(q-1)z^{q^2}.
\]
  \item $((\C_o)_{\{t_1\}})^\perp$ is a $[q^2,q^2-3,3]$ almost MDS code over $\gf(q)$.
  \item $((\C_o^\perp)_{\{t_1\}})^\perp$ is a $[q^2,4,q^2-q-1]$ code over $\gf(q)$ with weight enumerator
      \[
1+q^2(q-1)z^{q^2-q-1}+q(q^2-1)z^{q^2-q}+q^2(q-1)z^{q^2-1}+(q-1)z^{q^2}.
\]
  \item $(\C_o^\perp)_{\{t_1\}}$ is a $[q^2,q^2-4,4]$ almost MDS code over $\gf(q)$.
\end{itemize}
Furthermore, these codes hold 2-design.
\end{lemma}

The minimum linear locality of these punctured and shortened codes of ovoid codes and their duals are documented
in the following theorem.

\begin{theorem}\label{thm-localityps21}
Let $q \geq 4$.
Then the code  $(\C_o)_{\{t_1\}}$ is a $k$-optimal $(q^2,3,q^2-q,q;2)$-LLRC and the
code $((\C_o^\perp)_{\{t_1\}})^\perp$ is a $k$-optimal $(q^2,4,q^2-q-1,q;3)$-LLRC.
The code $(\C_o^\perp)_{\{t_1\}}$ is a $(q^2,q^2-4,4,q;q^2-q-2)$-LLRC and the code $((\C_o)_{\{t_1\}})^\perp$ is a
$(q^2,q^2-3,3,q; q^2-q-1)$-LLRC. Furthermore,
$(\C_o^\perp)_{\{t_1\}}$ and $((\C_o)_{\{t_1\}})^\perp$ are both $d$-optimal and $k$-optimal.
\end{theorem}

\begin{proof}
The parameters of these codes follow from Lemma \ref{lem-LDT20} and Corollary \ref{cor-newn}. It is easy to verify the distance optimality of $(\C_o^\perp)_{\{t_1\}}$ and $((\C_o)_{\{t_1\}})^\perp$ with respect to the Singleton-like bound.
The proofs of dimension optimality of $(\C_o^\perp)_{\{t_1\}}$ and $((\C_o)_{\{t_1\}})^\perp$ are similar, so we just prove the dimension optimality of $(\C_o^\perp)_{\{t_1\}}$. Taking $t=1$ into the right-hand side of the CM bound in \eqref{cm-bound}, one arrives at
\begin{align*}
  k & \leq r+k_{opt}^{(q)}(n-(r+1),d) \\
          & = (q-1)q-2+k_{opt}^{(q)}(q+1,4)\\
          & \leq q^2-4,
\end{align*}
where the last inequality holds due to the fact that $k_{opt}^{(q)}(q+1,3)\leq q-2$, which follows from the classical Singleton bound. Therefore, the code $(\C_o^\perp)_{\{t_1\}}$ is $k$-optimal. The completes the proof.
\end{proof}

Ovoid codes are very attractive in the sense that $\C_o^\perp$, $(\C_o^\perp)_{\{t_1\}}$ and $((\C_o)_{\{t_1\}})^\perp$
are both $d$-optimal and $k$-optimal. Recall that ovoid codes are the same as ovoids in projective geometry.
In addition, ovoid codes support
$3$-designs, which are related to inversive planes (also called M\"obius planes) \cite[Chapter 13]{Dingbook18}.
Furthermore, the trace codes of some ovoid codes are also optimal \cite{DH19}. These facts show that ovoid
codes are really diamonds.

\subsection{The minimum linear locality of maximal arc codes}

Throughout this section, let $q=2^m$ for some positive integer $m \geq 2$.
A \emph{maximal $(n, h)$-arc $\cA$}\index{maximal arc} in the projective plane
$\PG(2, \gf(q))$ is a subset of $n=hq+h-q$ points such that every line meets $\cA$
in $0$ or $h$ points. A maximal $(n, h)$-arc $\cA$ in $\PG(2, \gf(q))$ exists if and
only if $h$ divides $q$, where $2 \leq h <q$. Hence, in this section, we let $h=2^i$
for some $i$ with $1 \leq i <m$. There are several known families of maximal arcs
and the reader is referred to \cite[Section 12.7]{Dingbook18} for further information.

Let $\cA$ be a maximal $(n, h)$-arc  in PG$(2,\gf(q))$. Denote
$
\cA = \{\ba_1,\ba_2,\ldots,\ba_{n}\},
$
where each $\ba_i$ is a column vector in $\gf(q)^3$. Let $\C(\cA)$ denote the linear code over $\gf(q)$ with generator
matrix
$
G_\cA=[\ba_1\ba_2\cdots\ba_{n}].
$
We call  $\C(\cA)$ a \emph{maximum arc code}. The following theorem was proved in \cite[Section 12.7]{Dingbook18}.

\begin{theorem}\label{thm-marccode}
Let $q=2^m$ for any $m \geq 3$ and $h=2^i$ with $2 \leq i < m$. Let $\cA$ be a maximal
$(n, h)$-arc in $\PG(2, \gf(q))$. Then the maximum arc code $\C(\cA)$ has parameters
$[n,\, 3,\, n-h]$ and
weight enumerator
\begin{eqnarray}\label{eqn-wte23}
1+ \frac{(q^2-1)n}{h}z^{n-h} + \frac{(q^3-1)h-(q^2-1)n}{h}z^{n},
\end{eqnarray}
where $n=hq+h-q$. The dual code  $\C(\cA)^\perp$ has parameters $[n, n-3, 3]$.
Furthermore, the minimum weight codewords in both $\C(\cA)$ and $\C(\cA)^\perp$ support a $2$-design.
\end{theorem}

\begin{theorem}\label{thm-localitymaxarc}
Let $q=2^m$ for any $m \geq 3$ and $h=2^i$ with $2 \leq i < m$. Let $\cA$ be a maximal
$(n, h)$-arc in $\PG(2, \gf(q))$. Then $\C(\cA)$ is a $k$-optimal $(n, 3, n-h, q; 2)$-LLRC and
 $\C(\cA)^\perp$ is a $d$-optimal and $k$-optimal $(n, n-3, 3, q; n-h-1)$-LLRC.
\end{theorem}

\begin{proof}
It follows from Corollary \ref{cor-newn} and Theorem \ref{thm-marccode} that $\C(\cA)^\perp$ has minimum linear locality
$d(\C(\cA))-1$ and $\C(\cA)$ has  minimum linear locality $d(\C(\cA)^\perp)-1$. The parameters of the two codes then
follow from Theorem \ref{thm-marccode}. It is straightforward to verify that the parameters of $\C(\cA)^\perp$  meet
the Singleton-like bound.
We now prove the dimension optimality of $\C(\cA)^\perp$. Putting $t=1$ into the right-hand side of the CM bound in \eqref{cm-bound}, we have
\begin{align*}
  k & \leq r+k_{opt}^{(q)}(n-(r+1),d) \\
          & = n-h-1+k_{opt}^{(q)}(h,3)\\
          & \leq n-3,
\end{align*}
where the last inequality holds due to the fact that $k_{opt}^{(q)}(h,3)\leq h-2$, which follows from the classical Singleton bound. Therefore, the $\C(\cA)^\perp$ is $k$-optimal.

Taking the parameters of $\C(\cA)$ on the right-hand side of the CM bound in \eqref{cm-bound}, one has
\begin{align*}
  k & \leq \min_{t\in\Z_+}\{rt+k_{opt}^{(q)}(n-(r+1)t,d)\}\\
          & \leq r+k_{opt}^{(q)}(n-(r+1),d) \\
          & = 2+k_{opt}^{(q)}(n-3, n-h)\\
          & \leq 3,
\end{align*}
where the last inequality holds due to the fact that $k_{opt}^{(q)}(n-3,n-h)\leq 1$, which follows from the Plotkin bound. Therefore, $\C(\cA)$ is $k$-optimal. This completes the proof.
\end{proof}

A family of extended cyclic codes with the parameters of the code in Theorem \ref{thm-marccode} were documented in
\cite[Section 12.8]{Dingbook18}. We are interested in maximal arc codes, as they are $k$-optimal LLRCs  and their duals
are $d$-optimal and $k$-optimal LLRCs.

\section{The minimum linear locality of near MDS codes}\label{sec-localityofNMDS}

\subsection{Some general theory on the minimum linear locality of near MDS codes}

The Singleton defect of an $[n, k, d]$ code $\C$ is defined by def$(\C) = n- k + 1-d$. Thus, MDS codes are codes with defect 0. A code $\C$ is said to be almost MDS (AMDS for short) if it has defect 1. Hence, AMDS codes have parameters $[n, k, n-k]$. A code is said to be near MDS (NMDS for short) if the code and its dual code both are AMDS. By definition, $\C$ is near MDS if and only if $\C^\perp$ is so. Then an $[n, k]$ code $\C$ over $\gf(q)$ is NMDS if and
only if $d(\C)+d(\C^\perp) = n$ \cite{DL95}. The following lemma will be needed later (\cite{DodLan95}, \cite{FaldumWillems97}).

\begin{lemma}\label{lem-aug291}
Let $\C$ be an $[n, k, n-k]$ AMDS code over $\gf(q)$.
\begin{itemize}
\item If $k \geq 2$, then $\C$ is generated by its codewords of weight $n-k$ and $n-k+1$.
\item If $k \geq 2$ and $n-k >q$, then $\C$ is generated by its minimum weight codewords.
\end{itemize}
\end{lemma}

\begin{theorem}\label{thm-myJ2121}
Let $\C$ be a nontrivial NMDS code. Then the minimum linear locality of $\C$ is either $d(\C^\perp)-1$
or $d(\C^\perp)$. In particular, the minimum linear locality of $\C$ is $d(\C^\perp)-1$ if the minimum weight
codewords in $\C^\perp$ generate $\C^\perp$.
\end{theorem}

\begin{proof}
Let $n$ denote the length of $\C$. If $\C^\perp$ is generated by its minimum weight codewords,
it then follows from Theorem  \ref{thm-newJ1821} that $\C$ has minimum linear locality $d(\C^\perp)-1$.
Assume now that all the minimum weight codewords in $\C^\perp$ do not generate $\C^\perp$.
Since $\C^\perp$ is nontrivial, $\dim(\C^\perp) \geq 2$. It then follows from Lemma \ref{lem-aug291}
that $\C^\perp$ is generated by all the codewords of weights $d(\C^\perp)$ and $d(\C^\perp)+1$.
If the union of the supports of all the codes of weights $d(\C^\perp)$ and $d(\C^\perp)+1$ does not
contain $i  \in [n]$, then $\C$ must be have a zero coordinate in position $i$. This means that $d(\C)=1$,
which contradicts to our assumption that $\C$ is nontrivial.  It then follows from Theorem \ref{thm-generalLoc}
that the minimum linear locality of $\C$ is either $d(\C^\perp)-1$ or $d(\C^\perp)$.
\end{proof}

MDS codes are very interesting due to the following theorem whose proof is straightforward by following the assumptions and the parameters of NMDS codes.

\begin{theorem}\label{thm-nmdslocality0}
If a nontrivial NMDS code $\C$ over $\gf(q)$ with parameters $[n, k, n-k]$ has minimum linear locality $d(\C^\perp)-1$,
then $\C$
is a $d$-optimal and $k$-optimal $(n, k, n-k, q; k-1)$-LLRC with respect to the Singleton-like bound and the CM bound,
respectively.

If a nontrivial NMDS code $\C$ over $\gf(q)$ with parameters $[n, k, n-k]$ has minimum linear locality $d(\C^\perp)$,
then $\C$
is an almost $d$-optimal and $k$-optimal $(n, k, n-k, q; k)$-LLRC with respect to the Singleton-like bound and the CM bound,
respectively.
\end{theorem}

We will demonstrate later that some nontrivial NMDS codes $\C$ have minimum linear locality $d(\C^\perp)-1$
and some nontrivial NMDS codes $\C$ indeed have minimum linear locality $d(\C^\perp)$. Of course,
nontrivial NMDS codes $\C$ with minimum linear locality $d(\C^\perp)-1$ are better. Therefore, we are
more interested in nontrivial MDS code $\C$ with  minimum linear locality $d(\C^\perp)-1$.

\begin{corollary}\label{cor-21jan188}
Let $\C$ be a  nontrivial NMDS code over $\gf(q)$ with parameters $[n, k, n-k]$. If $\C^\perp$ does not have
a codeword of weight $d(\C^\perp)+1$, Then $\C$
is a $d$-optimal and $k$-optimal $(n, k, n-k, q; k-1)$-LLRC.
\end{corollary}

\begin{proof}
By definition, $\C^\perp$ has parameters $[n, n-k, k]$. Since $\C$ is nontrivial, $d(\C^\perp)=k \geq 2$. By Lemma
\ref{lem-aug291}, $\C^\perp$ is generated by its codewords of weights $d(\C^\perp)$ and $d(\C^\perp)+1$. Since  $\C^\perp$
does not have a codeword of weight $d(\C^\perp)+1$, $\C^\perp$ is generated by its codewords of weight $d(\C^\perp)$.
By Theorem \ref{thm-newJ1821}, $\C$ has minimum linear locality $d(\C^\perp)-1$. The desired conclusion then follows
from Theorem \ref{thm-nmdslocality0}.
\end{proof}

We remark that under the condition of Corollary \ref{cor-21jan188}, it can be proved that the minimum weight codewords
in $\C^\perp$ support a $1$-design \cite{DT20}.  To prove another result about the minimum linear locality, we need the following
lemma \cite{FaldumWillems97}.

\begin{lemma}\label{lem-121FW}
Let $\C$ be an NMDS code. Then for every minimum weight codeword $\bc$ in $\C$, there exists,
up to a multiple, a unique minimum weight codeword $\bc^\perp$ in $\C^\perp$ such that
$\support(\bc) \cap \support(\bc^\perp)=\emptyset$. In particular, $\C$ and $\C^\perp$
have the same number of minimum weight codewords.
\end{lemma}

We now provide the following result, which is useful in certain cases.

\begin{theorem}\label{thm-21jan191}
Let $\C$ be an NMDS code and let $d^\perp=d(\C^\perp)$. If
$$\bigcap_{S \in \cB_{d^\perp}(\C^\perp)} S = \emptyset,$$
then $\C^\perp$ has minimum linear locality $d(\C)-1$.
\end{theorem}

\begin{proof}
Let $\C$ have parameters $[n, k, d]$ with $d=n-k$. It follows from Lemma \ref{lem-121FW} that each set in
$\cB_d(\C)$ is the complement of a set in $\cB_{d^\perp}(\C^\perp)$ and vice versa. If an integer $i \in [n]$
is not in $\cup_{S \in \cB_d(\C)}$, we then deduce that it must be in $\cap_{S \in \cB_{d^\perp}(\C^\perp)}$.
This is contrary to the assumption that $\cap_{S \in \cB_{d^\perp}(\C^\perp)} S = \emptyset$. Consequently,
$$
\bigcup_{S \in \cB_d(\C)} = [n].
$$
The desired conclusion then follows from Corollary \ref{cor-newn0}.
\end{proof}

\subsection{The minimum linear locality of NMDS cyclic codes}

According to Corollary \ref{cor-newnc} and Theorem \ref{thm-nmdslocality0}, every nontrivial NMDS cyclic code $\C$ and its dual are both $d$-optimal
and $k$-optimal LLRCs. In this subsection, we document such NMDS cyclic codes. We begin with the following example.

\begin{example}
The ternary Golay code $\C_{\mathrm{Golay}}$ has parameters $[11, 6, 5]$ and weight enumerator
$$
1+132z^5+132z^6+330z^8+110z^9+24z^{11}.
$$
The dual code $\C_{\mathrm{Golay}}^\perp$ has parameters $[11, 5, 6]$ and weight enumerator
$$
1+132z^6+110z^9.
$$
Hence, the ternary Golay code is an nontrivial NMDS. It is well known that $\C_{\mathrm{Golay}}$ is a BCH code,
an irreducible cyclic code, and also a quadratic residue code. By Corollary \ref{cor-newnc} and Theorem \ref{thm-nmdslocality0}, both codes are
$d$-optimal and $k$-optimal LLRCs.
\end{example}

The following two theorems document two classes of nontrivial NMDS cyclic codes \cite{DT20}.

\begin{theorem}\label{thm-SQScode2}
Let $q=2^s$ with $s \geq 4$ being even. Then the narrow-sense BCH code $\C_{(q, q+1, 3,1)}$ over $\gf(q)$
has parameters $[q+1, q-3, 4]$, and its dual code $\C_{(q, q+1, 3,1)}^\perp$ has parameters $[q+1, 4, q-3]$
and weight enumerator
\begin{eqnarray*}
1+ \frac{(q-4)(q-1)q(q+1)}{24}z^{q-3} +  \frac{(q-1)q(q+1)}{2} z^{q-2} +\\
 \frac{(q+1)q^2(q-1)}{4} z^{q-1} +
\frac{(q-1)(q+1)(2q^2+q+6)}{6} z^q + \\
 \frac{3q^4 - 4q^3 - 3q^2 + 4q}{8} z^{q+1}.
\end{eqnarray*}
Further,  the codewords of weight $4$ in $\C_{(q, q+1, 3,1)}$ support a $2$-$(q+1, 4, (q-4)/2)$ design, and
the codewords of weight $q-3$ in the dual code $\C_{(q, q+1, 3,1)}^\perp$ support a $2$-$(q+1, q-3, \lambda^\perp)$ design
with
$$
\lambda^\perp=\frac{(q-4)^2(q-3)}{24}.
$$
\end{theorem}

\begin{theorem}\label{thm-SQScode}
Let $q=3^s$ with $s \geq 2$. Then the narrow-sense BCH code $\C_{(q, q+1, 3,1)}$ over $\gf(q)$
has parameters $[q+1, q-3, 4]$, and its dual code $\C_{(q, q+1, 3,1)}^\perp$ has parameters $[q+1, 4, q-3]$
and weight enumerator
\begin{eqnarray*}
1+ \frac{(q-1)^2q(q+1)}{24} z^{q-3} + \frac{(q-1)q(q+1)(q+3)}{4} z^{q-1} + \\
\frac{(q^2-1)(q^2-q+3)}{3} z^q +  \frac{3(q-1)^2q(q+1)}{8} z^{q+1}.
\end{eqnarray*}
Further, the minimum weight codewords in $\C_{(q, q+1, 3,1)}^\perp$ support a $3$-$(q+1, q-3, \lambda)$ design
with
$$
\lambda=\frac{(q-3)(q-4)(q-5)}{24},
$$
and the minimum weight codewords in $\C_{(q, q+1, 3,1)}$ support a $3$-$(q+1, 4, 1)$ design, i.e.,
a Steiner quadruple system\index{Steiner quadruple system} $S(3,4,3^s+1)$. Furthermore,
the codewords of weight 5 in $\C_{(q, q+1, 3,1)}$ support a $3$-$(q+1, 5, (q-3)(q-7)/2)$ design.
\end{theorem}

Notice that the minimum weight codewords in the code $\C_{(q, q+1, 3,1)}^\perp$ in Theorem \ref{thm-SQScode}
generate the code, as it does not have a codeword of weight $q-2$.
The next theorem gives the third class of NMDS cyclic codes \cite{TD21}.

\begin{theorem}\label{thm:odd-4designs}
Let $m \geq 5$ be odd and $q=2^m$. Then the narrow-sense BCH code $\C_{(q, q+1, 4,1)}$ over $\gf(q)$
is a $[q+1, q-5, 6]$ NMDS code and $\C_{(q, q+1, 4,1)}^\perp$  has parameters $[q+1, 6, q-5]$.
Furthermore,
the minimum weight codewords in $\C_{(q, q+1, 4,1)}$ support a $4$-$(q+1, 6, (q-8)/2)$ design
and
the minimum weight codewords in $\C_{(q, q+1, 4,1)}^\perp$ support a $4$-$(q+1, q-5, \lambda)$ design
with
$$
\lambda=\frac{q-8}{30} \binom{q-5}{4}.
$$
\end{theorem}

Many NMDS codes have been constructed (see, for example, \cite{AL05,AL08,DeBoer96,DL95,DodLan95,DodLan00,FaldumWillems97,Giulietti04,JK19,MMP02,TongDing,WH21,TV91}). It is worthwhile to check if some of them
are cyclic.

\subsection{The minimum linear locality of the extended codes of some NMDS cyclic codes}

In this section, we investigate the minimum linear locality of the extended codes of some NMDS cyclic codes, and will make
use of Theorem \ref{thm-loca-ext-cycliccode}.

\begin{theorem}\label{thm-SQScodeext}
Let $q=3^s$ with $s \geq 2$. Then the extended code $\overline{\C_{(q, q+1, 3,1)}}$ over $\gf(q)$
has parameters $[q+2, q-3, 5]$, and its dual code $(\overline{\C_{(q, q+1, 3,1)}})^\perp$ has parameters
$[q+2, 5, q-3]$. Furthermore, $(\overline{\C_{(q, q+1, 3,1)}})^\perp$ is a $d$-optimal and $k$-optimal
$(q+2, 5, q-3, q; 4)$-LLRC.
\end{theorem}

\begin{proof}
Put $n=q+1$.
Let $\alpha$ be a generator of $\gf(q^2)^*$ and $\beta=\alpha^{q-1}$. Then $\beta$ is an $n$-th root of unity
in $\gf(q^2)$. Let $\m_\beta(x)$ and $\m_{\beta^2}(x)$ denote the minimal polynomial of $\beta$ and $\beta^2$ over $\gf(q)$,
respectively.
Note that $\m_\beta(x)$ has only roots $\beta$ and $\beta^q$ and  $\m_{\beta^2}(x)$ has roots $\beta^2$ and $\beta^{q-1}$.
We deduce that $\m_\beta(x)$ and $\m_{\beta^2}(x)$ are distinct irreducible polynomials of degree $2$.
By definition, $g(x):=\m_\beta(x)\m_{\beta^2}(x)$ is the generator polynomial of $\C_{(q, q+1, 3,1)}$.  Put $\gamma=\beta^{-1}$. Then $\gamma^{q+1}=\beta^{-(q+1)}=1$.
It then follows from Delsarte's theorem that the trace expression of $\C_{(q, q+1, 3,1)}^\perp$ is given by
\begin{eqnarray*}
\C_{(q, q+1, 3,1)}^\perp=\{\bc_{(a,b)}: a, b \in \gf(q^2)\},
\end{eqnarray*}
where $\bc_{(a,b)}=(\tr_{q^2/q}(a\gamma^i +b \gamma^{2i}))_{i=0}^q$. Define
\begin{eqnarray*}
H=\left[
\begin{array}{rrrrr}
1  & \gamma^1 & \gamma^2 & \cdots & \gamma^q \\
1  & \gamma^2 & \gamma^4 & \cdots & \gamma^{2q}
\end{array}
\right].
\end{eqnarray*}
It is easily seen that $H$ is a parity-check matrix of $\C_{(q, q+1, 3,1)}$, i.e.,
$$
\C_{(q, q+1, 3,1)}=\{\bc \in \gf(q)^{q+1}: \bc H^T=\bzero\}.
$$
Note that $\{1, \gamma\}$ is a basis of $\gf(q^2)$ over $\gf(q)$. Every $\gamma^j$ can be expressed as
$\gamma^j=a_{j,0}+a_{j,1}\gamma$, where $a_{j,0} \in \gf(q)$ and $a_{j,1} \in \gf(q)$. Later, $H$ refers to
the $4 \times n$ matrix over $\gf(q)$ defined by these $a_{j,i}$.

From Lemma \ref{thm-extendedCodeParam}, $\overline{\C_{(q, q+1, 3,1)}}$ has parameters $[q+2,q-3,d(\overline{\C_{(q, q+1, 3,1)}})]$ and the parity-check matrix $\overline{H}$ of $\overline{\C_{(q, q+1, 3,1)}}$ is
\begin{eqnarray*}
\overline{H}=\left[
\begin{array}{ll}
\bone & 1 \\
H &     \bzero
\end{array}
\right],
\end{eqnarray*}
where $\bone=(1,1, \ldots, 1)$ and $\bzero=(0,0, \ldots, 0)^T$.

Note that $d(\overline{\C_{(q, q+1, 3,1)}})=d(\C_{(q, q+1, 3,1)})+1$ or $d(\overline{\C_{(q, q+1, 3,1)}})=d(\C_{(q, q+1, 3,1)})$. Now we prove that $d(\overline{\C_{(q, q+1, 3,1)}})=d(\C_{(q, q+1, 3,1)})+1=5$. Let $U_{q+1}$ denote the set of all $(q+1)$-th roots of unity in $\gf(q^2)$. Suppose $d(\overline{\C_{(q, q+1, 3,1)}})=d(\C_{(q, q+1, 3,1)})=4$. Then there are four pairwise distinct elements $x, y, z, w$ in $U_{q+1}$ such that
\begin{eqnarray}\label{eqn-tchm105}
a \left[ \begin{array}{c}
1 \\
x \\
x^2
\end{array}
\right]
+
b \left[ \begin{array}{c}
1 \\
y \\
y^2
\end{array}
\right]
+
c \left[ \begin{array}{c}
1 \\
z \\
z^2
\end{array}
\right]
+
d \left[ \begin{array}{c}
1 \\
w \\
w^2
\end{array}
\right]
=0,
\end{eqnarray}
where $a, b, c,d \in \gf(q)^*$.
Raising to the $q$-th power both sides of the equation $ax+by+cz+dw=0$ yields
\begin{eqnarray}\label{eqn-tchm106}
ax^{-1}+by^{-1}+cz^{-1}+dw^{-1}=0.
\end{eqnarray}
Combining (\ref{eqn-tchm105}) and (\ref{eqn-tchm106}) gives
\begin{eqnarray*}
a \left[ \begin{array}{c}
x^{-1} \\
1 \\
x \\
x^2
\end{array}
\right]
+
b \left[ \begin{array}{c}
y^{-1} \\
1 \\
y \\
y^2
\end{array}
\right]
+
c \left[ \begin{array}{c}
z^{-1} \\
1 \\
z \\
z^2
\end{array}
\right]
+
d \left[ \begin{array}{c}
w^{-1} \\
1 \\
w \\
w^2
\end{array}
\right]
=0.
\end{eqnarray*}
It then follows that
\begin{eqnarray*}
\left|
\begin{array}{llll}
x^{-1} & y^{-1} & z^{-1} & w^{-1} \\
1 & 1 & 1 & 1 \\
x  & y & z & w \\
x^2 & y^2 & z^2 & w^{2}
\end{array}
\right|=\frac{(x-y)(x-z)(x-w)(y-z)(y-w)(z-w)}{xyzw}
=0.
\end{eqnarray*}
This is contrary to our assumption that $x, y, z,w$ are pairwise distinct. Hence,
\begin{eqnarray}\label{eqn-21J25}
d(\overline{\C_{(q, q+1, 3,1)}})=d(\C_{(q, q+1, 3,1)})+1=5.
\end{eqnarray}

We now prove the following equalities:
\begin{eqnarray}\label{eqn-2021jan25}
(\overline{\C_{(q, q+1, 3,1)}})^\perp = \widetilde{ \overline{\C_{(q, q+1, 3,1)}^\perp} },
\end{eqnarray}
where $\widetilde{D}$ denotes the augmented code of a code $D$. It is easily verified that the sum of all coordinates in each codeword
 $$
 \bc_{(a,b)}=(\tr_{q^2/q}(a\gamma^i +b \gamma^{2i}))_{i=0}^q
 $$
 is zero, as both $\gamma$ and $\gamma^2$ are $n$-th roots of unity. Consequently,
 $\overline{\C^\perp}$ is generated by the matrix $[H \bzero]$, where $H$ is the $4 \times n$ matrix over $\gf(q)$ defined above. Then the equality in  (\ref{eqn-2021jan25}) follows from Lemma  \ref{thm-extendedCodeParam}.
 By (\ref{eqn-21J25}), we conclude that the all-one vector is not a codeword in $\C_{(q, q+1, 3,1)}^\perp$.
It then follows from (\ref{eqn-2021jan25}) that
$$
\dim ((\overline{\C_{(q, q+1, 3,1)}})^\perp) = 1+ \dim(\C_{(q, q+1, 3,1)}^\perp) = 5.
$$

Now we prove $d((\overline{\C_{(q, q+1, 3,1)}})^\perp)=q-3$. Note that $(\overline{\C_{(q, q+1, 3,1)}})^\perp$ has generator matrix $\overline{H}$ and
\begin{eqnarray*}
\C_{(q, q+1, 3,1)}^\perp=\{\bc_{(a,b)}: a, b \in \gf(q^2)\},
\end{eqnarray*}
where $\bc_{(a,b)}=(\tr_{q^2/q}(a\gamma^i +b \gamma^{2i}))_{i=0}^q$.
It follows from (\ref{eqn-2021jan25}) that
the codewords in $(\overline{\C_{(q, q+1, 3,1)}})^\perp$ have the form $(\bc_{(a,b)}+c\bone,c)$, where $c\in \gf(q)$.
Let $u \in U_{q+1}$. Then
\begin{eqnarray*}
\tr_{q^2/q}(au+bu^2)+c
&=& au+bu^2 + a^qu^{-1}+b^q u^{-2}+c \\
&=& u^{-2}(bu^4+au^3 +a^qu+b^q+cu^2).
\end{eqnarray*}
Hence, there are at most four $u \in U_{q+1}$ such that $\tr_{q^2/q}(au+bu^2)+c=0$ if $(a, b, c) \neq (0,0,0)$.
As a result, for $(a, b, c) \neq (0,0,0)$,
we have
$$
\wt((\bc_{(a,b)}+c\bone,c))=\wt(\bc_{(a,b)}+c\bone)+1 \geq q+1-4+1=q-2,
$$
and for $(a, b) \neq (0,0), c=0$, we have
$$
\wt((\bc_{(a,b)},0))=\wt(\bc_{(a,b)}) \geq q+1-4=q-3.
$$
This means that $d((\overline{\C_{(q, q+1, 3,1)}})^\perp)\geq q-3$. If $d((\overline{\C_{(q, q+1, 3,1)}})^\perp)=q-2$, then $(\overline{\C_{(q, q+1, 3,1)}})^\perp$ would be an MDS code and $\overline{\C_{(q, q+1, 3,1)}}$ would also be an MDS code, which leads to a contradiction. We then conclude that $d((\overline{\C_{(q, q+1, 3,1)}})^\perp)=q-3$. Now both $\overline{\C_{(q, q+1, 3,1)}}$ and its dual are AMDS. Since $d(\overline{\C_{(q, q+1, 3,1)}})=5=d(\C_{(q, q+1, 3,1)})+1$, by Theorem \ref{thm-loca-ext-cycliccode} we deduce that
$(\overline{\C_{(q, q+1, 3,1)}})^\perp$ has locality $d(\C_{(q, q+1, 3,1)})=4$.
The optimality of $(\overline{\C_{(q, q+1, 3,1)}})^\perp$ then
follows  from Theorem \ref{thm-nmdslocality0}.
\end{proof}

\begin{theorem}\label{thm-2SQScodeext}
Let $q=2^s$ with $s \geq 4$ being even. Then the extended code $\overline{\C_{(q, q+1, 3,1)}}$ over $\gf(q)$
has parameters $[q+2, q-3, 5]$, and its dual code $(\overline{\C_{(q, q+1, 3,1)}})^\perp$ has parameters
$[q+2, 5, q-3]$. Furthermore, $(\overline{\C_{(q, q+1, 3,1)}})^\perp$ is a $d$-optimal and $k$-optimal
$(q+2, 5, q-3, q; 4)$-LLRC.
\end{theorem}

\begin{proof}
The proof of this theorem is similar to that of Theorem \ref{thm-SQScodeext} and is omitted.
\end{proof}

Notice that the two theorems above provide not only two families of $d$-optimal and $k$-optimal LLRCs, but also
two families of NMDS codes with new parameters.

\subsection{The linear locality of some NMDS codes from oval polynomials}

Oval polynomials were used to construct NMDS codes in \cite{WH21}. These NMDS codes are not cyclic.
In this subsection, we study the minimum linear locality of some of them. To introduce these codes, we
need oval polynomials. Throughout this subsection, let $q=2^m$, where $m \geq 3$.

An oval polynomial $f(x)$ on $\gf(q)$ is a polynomial such that
\begin{itemize}
\item  $f$ is a permutation polynomial of $\gf(q)$ with $\deg(f)<q$ and $f(0)=0$, $f(1)=1$;  and
\item for each $a \in \gf(q)$, $g_a(x):=(f(x+a)+f(a))x^{q-2}$ is also a permutation polynomial
      of $\gf(q)$.
\end{itemize}
Every oval polynomial $f$ can be used to construct a hyperoval in $\PG(2, \gf(q))$ \cite[Chapter 12]{Dingbook18}.
The following is a list of known infinite families of oval polynomials in the literature.

\begin{theorem}\label{thm-knownopolys}
Let $m \geq 4$ be an integer. The following are oval polynomials of $\gf(q)$, where $q=2^m$.
\begin{itemize}
\item The translation polynomial $f(x)=x^{2^h}$, where $\gcd(h, m)=1$.
\item The Segre polynomial $f(x)=x^6$, where $m$ is odd.
\item The Glynn oval polynomial $f(x)=x^{3 \times 2^{(m+1)/2} +4}$, where $m$ is odd.
\item The Glynn oval polynomial $f(x)=x^{ 2^{(m+1)/2} + 2^{(m+1)/4} }$ for $m \equiv 3 \pmod{4}$.
\item The Glynn oval polynomial $f(x)=x^{ 2^{(m+1)/2} + 2^{(3m+1)/4} }$ for $m \equiv 1 \pmod{4}$.
\item The Cherowitzo oval polynomial $f(x)=x^{2^e}+x^{2^e+2}+x^{3 \times 2^e+4},$ where $e=(m+1)/2$ and $m$ is odd.
\item The Payne oval polynomial $f(x)=x^{\frac{2^{m-1}+2}{3}} + x^{2^{m-1}} + x^{\frac{3 \times 2^{m-1}-2}{3}}$,
        where $m$ is odd.
\item The Subiaco polynomial
$$
f_a(x)=((a^2(x^4+x)+a^2(1+a+a^2)(x^3+x^2)) (x^4 + a^2 x^2+1)^{2^m-2}+x^{2^{m-1}},
$$
where $\tr_{q/2}(1/a)=1$ and $a \not\in \gf(4)$ if $m \equiv 2 \bmod{4}$.
\item The Adelaide oval polynomial
$$
f(x)=\frac{T(\beta^m)(x+1)}{T(\beta)} + \frac{T((\beta x + \beta^q)^m)}{T(\beta) (x+T(\beta)x^{2^{m-1}} +1)^{m-1}} + x^{2^{m-1}},
$$
where $m \geq 4$ is even, $\beta \in \gf(q^2) \setminus \{1\}$ with $\beta^{q+1}=1$, $m \equiv \pm (q-1)/3 \pmod{q+1}$,
and $T(x)=x+x^q$.
\end{itemize}
\end{theorem}

The following lemma will be needed later \cite{Masch98}.

\begin{lemma}\label{lem-opoly2to1}
A polynomial $f$ over $\gf(q)$ with $f(0)=0$ is an oval polynomial if and only if $f_u:=f(x)+ux$
is $2$-to-$1$ for every $u \in \gf(q)^*$.
\end{lemma}

\subsubsection{NMDS codes with parameters $[q+3, 3, q]$ from oval polynomials}

Let $f$ be a polynomial over $\gf(q)$ with $f(0)=0$ and $f(1)=1$. Let $\alpha$ be a generator of $\gf(q)^*$.
Define
\begin{eqnarray*}
\bar{B}_f=\left[
\begin{array}{llllllll}
f(0) & f(\alpha^0)  & f(\alpha^1) & \cdots & f(\alpha^{q-2}) & 1 & 0 & 1\\
0 &\alpha^0  & \alpha^1 & \cdots & \alpha^{q-2} & 0 & 1 & 1\\
1      & 1         &   1           & \cdots & 1                    & 0 & 0 & 0
\end{array}
\right].
\end{eqnarray*}
Let $\bar{\C}_f$ denote the linear code over $\gf(q)$ with generator matrix $\bar{B}_f$. The following theorem was proved in
\cite{WH21}.

\begin{theorem}\label{thm-J271}
Let $m \geq 3$, and let $f$ be an oval polynomial over $\gf(q)$. Then the code $\bar{\C}_f$ is an NMDS code
over $\gf(q)$ with parameters $[q+3, 3, q]$ and weight enumerator
\begin{eqnarray*}
1+\frac{(q-1)(q+2)}{2} z^q + \frac{(q-1)q(q+2)}{2}z^{q+1} +\frac{(q-1)q}{2} z^{q+2} + \frac{(q-2)(q-1)q}{2} z^{q+3}.
\end{eqnarray*}
\end{theorem}

\begin{theorem}\label{thm-21021-211}
The dual code $(\bar{\C}_f)^\perp$ is a $d$-optimal and $k$-optimal $(q+3, q, 3, q; q-1)$-LLRC.
\end{theorem}

\begin{proof}
Let $\bc_1$, $\bc_2$ and $\bc_3$ denote the first, second and third rows of the generator matrix $\bar{B}_f$,
respectively. By the definition of the
polynomial $f$, it is easily seen that $\bc_1+\bc_3$ and $\bc_2+\bc_3$ are two minimum weight codewords in
$\bar{\C}_f$. In addition, the supports of these two codewords are $[q+3] \setminus \{q+1\}$ and
$[q+3] \setminus \{q\}$, respectively. Clearly,
$$
([q+3] \setminus \{q+1\}) \cup ([q+3] \setminus \{q\}) = [q+3].
$$
By Corollary \ref{cor-newn0},  $(\bar{\C}_f)^\perp$ has minimum linear locality $d(\bar{\C}_f)-1=q-1$.
The desired conclusion then follows from Theorem \ref{thm-nmdslocality0}.
\end{proof}

The minimum linear locality of $\bar{\C}_f$ is given below.

\begin{theorem}\label{thm-2345}
The NMDS code $\bar{\C}_f$ is a $d$-optimal and $k$-optimal $(q+3, 3, q, q; 2)$-LLRC.
\end{theorem}

\begin{proof}
Recall we use the elements in the set $[q+3]=\{0,1, \ldots, q+2\}$ to index the coordinate positions
of the code  $\bar{\C}_f$ and its dual.
Since all the codewords of weight 3 in $(\bar{\C}_f)^\perp$ were characterised in \cite{WH21},
we outline a proof here only. Notice that the union of the supports of all the codewords
of weight 3 in $(\bar{\C}_f)^\perp$ specified in Case 1 of the proof of Theorem 8 in \cite{WH21} is $\{0,1, \ldots, q-1, q+2\}$,
and  the union of the supports of all the codewords
of weight 3 in $(\bar{\C}_f)^\perp$ specified in Case 4 of the proof of Theorem 8 in \cite{WH21} is $\{q,q+1,  q+2\}$.
It then follows that
$$
\bigcup_{S \in \cB_3((\bar{\C}_f)^\perp)} S =[q+3].
$$
By Corollary \ref{cor-newn0},  $(\bar{\C}_f)$ has minimum linear locality $d((\bar{\C}_f)^\perp) -1=2$.
The desired conclusion then follows from Theorem \ref{thm-nmdslocality0}.
\end{proof}

We remark that the NMDS code $(\bar{\C}_f)$ is an extended hyperoval code
(see for example \cite[Section 12.2]{Dingbook18}). The reader is referred to \cite{WH21} for detail.

\subsubsection{NMDS codes with parameters $[q+1, 3, q-2]$ from oval polynomials}

Let $f$ be a polynomial over $\gf(q)$ with $f(0)=0$ and $f(1)=1$. Let $\alpha$ be a generator of $\gf(q)^*$.
Define
\begin{eqnarray}
G_f=\left[
\begin{array}{llllll}
f(\alpha^0)  & f(\alpha^1) & \cdots & f(\alpha^{q-2}) & 0 & 1 \\
\alpha^0  & \alpha^1 & \cdots & \alpha^{q-2} & 1 & 0 \\
1               &   1           & \cdots & 1                    & 1 & 1
\end{array}
\right].
\end{eqnarray}
Let $\C_f$ denote the linear code over $\gf(q)$ with generator matrix $G_f$. The following result was proved in \cite{WH21}.

\begin{theorem}\label{thm-nmdscodej231}
Let $m \geq 3$ be odd and let $f(x)$ be an oval polynomial over $\gf(q)$ with coefficients in $\gf(2)$. Then $\C_f$ is a $[q+1, 3, q-2]$ NMDS code
over $\gf(q)$ with weight enumerator
\begin{eqnarray*}
A(z)=1 + (q-1)(q-2)z^{q-2} + \frac{(q-1)(q^2-5q+12)}{2} z^{q-1} + \\
(q-1)(4q-5) z^{q} + \frac{(q-1)(q^2-3q+4)}{2} z^{q+1}.
\end{eqnarray*}
\end{theorem}

This class of NMDS codes are very important to us, as they demonstrate that some nontrivial NMDS codes $\C$ indeed
have minimum linear locality $d(\C^\perp)$ rather than $d(\C^\perp)-1$. We will need the following lemma later.

\begin{lemma}\label{lem-Jan2121}
Let $m \geq 3$ be odd and let $f(x)$ be an oval polynomial over $\gf(q)$ with coefficients in $\gf(2)$.
Then $f(x)+x+1=0$ does not have a solution $x \in \gf(q)$.
\end{lemma}

\begin{proof}
By the definition of oval polynomials, $0$ and $1$ are not not solutions of the equation $f(x)+x+1=0$.
Suppose that $f(a)+a+1=0$ for some $a \in \gf(q)\setminus \{0,1\}$. Since $f(x)$ has coefficients
in $\gf(2)$, we have  $f(a^2)+a^2+1=0$ and $f(a^4)+a^4+1=0$. It then follows from  Lemma \ref{lem-opoly2to1}
that the set $\{a, a^2, a^4\}$ has cardinality at most $2$.
If $a=a^2$, then $a \in \{0,1\}$, which contradicts to the assumption that  $a \in \gf(q)\setminus \{0,1\}$.
If $a=a^4$, then $a=0$ or $a^3=1$. If $a^3=1$, we have $a=1$ as $\gcd(3, q-1)=1$.
Hence,  $a=a^4$ implies that $a \in \{0,1\}$,
which contradicts to the assumption that  $a \in \gf(q)\setminus \{0,1\}$.
If $a^2=a^4$, then $a \in \{0,1\}$, which contradicts to the assumption that  $a \in \gf(q)\setminus \{0,1\}$.
This completes the proof.
\end{proof}

\begin{theorem}\label{thm-bestexamthm}
Let $m \geq 3$ be odd and let $f(x)$ be an oval polynomial over $\gf(q)$ with coefficients in $\gf(2)$, and let $\C_f$ be
the code in Theorem \ref{thm-nmdscodej231}. Then $\C_f^\perp$ has minimum linear locality $d(\C)-1$ and is a
$d$-optimal and $k$-optimal $(q+1, q-2, 3, q; q-3)$-LLRC, and $\C_f$ has minimum linear locality $d(\C_f^\perp)$
and is an almost $d$-optimal and $k$-optimal $(q+1,3,q-2,q; 3)$-LLRC.
\end{theorem}

\begin{proof}
Since all the minimum weight codewords in $\C_f$ were not characterized in \cite{WH21}, we have to do this job here.
Let $\bv_1$, $\bv_2$ and $\bv_3$ denote the first, second and third rows in the generator matrix $G_f$ above.

We first first consider all the codewords $\bv_3+\bv_2+b\bv_1$, where $b \in \gf(q)^*\setminus \{1\}$. By definition,
$$
\bv_3+\bv_2+b\bv_1=
(b, 1+\alpha^1+bf(\alpha^1), \ldots,  1+\alpha^{q-2}+bf(\alpha^{q-2}), 0, 1+b).
$$
For any $i$ with $1 \leq i \leq q-2$, put $b=(1+\alpha^i)/f(\alpha^i)$. It then follows from Lemma \ref{lem-Jan2121}
that $b \neq 1$. Clearly, $b \neq 0$, as $\alpha$ is a generator of $\gf(q)$. It then follows from Lemma  \ref{lem-opoly2to1}
that there is a unique $j \in \{1, 2, \ldots, q-2\} \setminus \{i\}$ such that $b=(1+\alpha^j)/f(\alpha^j)$. Consequently,
$$
\wt(\bv_3+\bv_2+b\bv_1)=q+1-3=q-2.
$$
The support of this codeword is
$$
\support(\bv_3+\bv_2+b\bv_1)=[q+1]\setminus \{i, j, q-1\}.
$$
The total number of such choices of $b$ is $(q-2)/2$. Since $\bv_3$ is the all-one codeword, we have already  characterized
$(q-1)(q-2)/2$ minimum weight codewords of this form in $\C_f$.

We then consider  all the codewords $\bv_3+a\bv_2+\bv_1$, where $a \in \gf(q)^*\setminus \{1\}$. By definition,
$$
\bv_3+a\bv_2+\bv_1=
(a, 1+a\alpha^1+f(\alpha^1), \ldots,  1+a  \alpha^{q-2}+f(\alpha^{q-2}), 1+a, 0).
$$
For any $i$ with $1 \leq i \leq q-2$, put $a=(1+f(\alpha^i))/\alpha^i$.
It then follows from Lemma \ref{lem-Jan2121}
that $a \neq 1$. Clearly, $a \neq 0$, as $\alpha$ is a generator of $\gf(q)$ and $f(1)=1$. It then follows from Lemma
\ref{lem-opoly2to1}
that there is a unique $j \in \{1, 2, \ldots, q-2\} \setminus \{i\}$ such that $b=(1+f(\alpha^j))/\alpha^j$. Consequently,
$$
\wt(\bv_3+a\bv_2+\bv_1)=q+1-3=q-2.
$$
The support of this codeword is
$$
\support(\bv_3+a\bv_2+\bv_1)=[q+1]\setminus \{i, j, q\}.
$$
The total number of such choices of $a$ is $(q-2)/2$. Since $\bv_3$ is the all-one codeword, we have already characterized
 $(q-1)(q-2)/2$ minimum weight codewords of this form in $\C_f$.

In the first case above, the coordinate in position $q-1$ in the codewords $\bv_3+\bv_2+b\bv_1$ is  zero and the coordinate
in position $q$ in these codewords is nozero.
In the second case above, the coordinate in position $q-1$ in the codewords $\bv_3+a\bv_2+\bv_1$ is nonzero and the coordinate in position $q$ in these codewords is zero. Therefore, the minimum weight codewords in the two forms do not
overlap. By Theorem \ref{thm-nmdscodej231}, we have characterized all the minimum weight codewords in $\C_f$. From
the discussions above, we have
$$
\bigcup_{S \in \cB_{q-2}(\C_f)} S= [q+1].
$$
It then follows from Corollary \ref{cor-newn0} that $\C_f^\perp$ has minimum locality $d(\C_f)-1$.

The discussions above showed that the coordinate in position $0$ in all the minimum weight codewords in $\C_f$ is nonzero.
It then follows from  Lemma \ref{lem-121FW} that  the coordinate in position $0$ in all the minimum weight codewords in $\C_f^\perp$ is zero. This means that
$$
0 \not\in \bigcup_{S \in \cB_{3}(\C_f^\perp)} S.
$$
Hence, $\C_f$ does not have minimum linear locality $d(\C_f^\perp)-1$. It then follows from Theorem \ref{thm-myJ2121}
that $\C_f$ has minimum linear locality $d(\C_f^\perp)$. The remaining desired conclusions then follow from
Theorems    \ref{thm-nmdscodej231}  and \ref{thm-nmdslocality0}.
\end{proof}

The proof of Theorem \ref{thm-bestexamthm} shows that it could be hard to determine the minimum linear locality
of an NMDS code.

\subsubsection{NMDS codes with parameters $[q+2, 3, q-1]$ from oval polynomials}

Let $f$ be a polynomial over $\gf(q)$ with $f(0)=0$ and $f(1)=1$. Let $\alpha$ be a generator of $\gf(q)^*$.
Define
\begin{eqnarray}
\bar{G}_f=\left[
\begin{array}{lllllll}
f(0) &f(\alpha^0)  & f(\alpha^1) & \cdots & f(\alpha^{q-2}) & 0 & 1 \\
0 &\alpha^0  & \alpha^1 & \cdots & \alpha^{q-2} & 1 & 0 \\
1 &1               &   1           & \cdots & 1                    & 1 & 1
\end{array}
\right].
\end{eqnarray}
By definition, $\bar{G}_f$ is a $3$ by $q+2$ matrix over $\gf(q)$. Let $\bar{\C}_f$ denote the linear code over $\gf(q)$ with
generator matrix $\bar{G}_f$.

The following theorem was presented in \cite{WH21}, but its proof was omitted in \cite{WH21}.

\begin{theorem}\label{thm-nmdscodej2322}
Let $m \geq 3$ be odd and let $f(x)$ be an oval polynomial over $\gf(q)$ with coefficients in $\gf(2)$. Then $\bar{\C}_f$ is a $[q+2, 3, q-1]$ NMDS code
over $\gf(q)$ with weight enumerator
\begin{eqnarray*}
\bar{A}(z)=1 + (q-1)(q-2)z^{q-1} + \frac{(q-1)(q^2-3q+14)}{2} z^{q} + \\
3(q-1)(q-2) z^{q+1} + \frac{(q-1)(q^2-3q+4)}{2} z^{q+2}.
\end{eqnarray*}
\end{theorem}

\begin{theorem}\label{thm-bestexamthm2}
Let $m \geq 3$ be odd and let $f(x)$ be an oval polynomial over $\gf(q)$ with coefficients in $\gf(2)$, and let $\bar{\C}_f$ be
the code in Theorem \ref{thm-nmdscodej2322}. Then $\bar{\C}_f^\perp$ has minimum linear locality $d(\C_f)-1$ and is a
$d$-optimal and $k$-optimal $(q+2, q-1, 3, q; q-2)$-LLRC, and $\bar{\C}_f$ has minimum linear locality $d(\bar{\C}_f^\perp)$
and is an almost $d$-optimal and $k$-optimal $(q+2,3,q-1,q; 3)$-LLRC.
\end{theorem}

\begin{proof}
The proof is similar to that of Theorem \ref{thm-bestexamthm} and is omitted here. However, we inform the reader that
the coordinates in positions $0$ and $1$ in all the minimum weight codewords in $\bar{\C}_f$ are always nonzero. This means
that
$$
\{0,1\} \bigcap  \bigcup_{S \in \cB_{3}(\bar{\C}_f^\perp)} S = \emptyset.
$$
To prove this theorem, one has to characterize all the minimum weight codewords in $\bar{\C}_f$.
\end{proof}

Notice that Theorem \ref{thm-bestexamthm} documents the second class of nontrivial linear codes $\C$ with minimum linear locality more than $d(\C^\perp)-1$.

\subsection{The minimum linear locality of other NMDS codes}

The following result says that other infinite families of NMDS codes do exist.

\begin{theorem}[\cite{TV91}]
Algebraic geometric $[n, k, n-k]$ NMDS codes over $\gf(q)$, $q=p^m$, do exist for every $n$ with
\begin{eqnarray*}
n \leq
\left\{
\begin{array}{ll}
q + \lceil 2 \sqrt{q} \rceil & \mbox{ if $p$ divides $\lceil 2 \sqrt{q} \rceil$ and $m$ is odd,} \\
q + \lceil 2 \sqrt{q} \rceil +1 & \mbox{ otherwise,}
\end{array}
\right.
\end{eqnarray*}
and arbitrary $k \in \{2,3, \ldots, n-2\}$.
\end{theorem}

Many NMDS codes have been constructed (see, for example, \cite{AL05,AL08,DeBoer96,DL95,DodLan95,DodLan00,FaldumWillems97,Giulietti04,JK19,MMP02,TongDing,WH21,TV91}). It is valuable to check which of the NMDS codes $\C$ have minimum linear locality $d(\C^\perp)-1$, as they are $d$-optimal and $k$-optimal LLRCs.
It has been observed that the analysis of the minimum linear locality of a linear code is harder than the determination of the minimum distance of the dual code.

\section{Summary and concluding remarks}\label{sec-final}

The objectives of this paper are to develop some general theory for the minimum linear locality of linear codes and search for $d$-optimal
or $k$-optimal LLRCs in known families of linear codes. Below is a summary of the major general results on the minimum linear locality of nontrivial linear codes developed in this paper.
\begin{enumerate}
\item We proved that every nontrivial linear code has a minimum linear locality and showed how to find it
(see Theorem \ref{thm-generalLoc}).
\item We gave a necessary and sufficient condition for a nontrivial linear code $\C$ to have the minimum linear
locality $d(\C^\perp)-1$ (see Corollary \ref{cor-newn0}).
\item We determined the minimum linear locality for nontrivial linear codes $\C$ such that the minimum weight
codewords in $\C^\perp$ support a $1$-design
(see Corollary \ref{cor-newn}). This general result has settled the minimum linear locality of many families of linear codes,
as many families of linear codes support $1$-designs.
\item We determined the minimum linear locality for nontrivial linear codes $\C$ whose automorphism group is transitive
(see Corollary \ref{cor-newtrans}).  This general result has settled the minimum linear locality of many families of linear codes,
including the affine-invariant linear codes.
\item We determined the minimum linear locality of $(\overline{\C})^\perp$ under the condition that $d(\overline{\C})=d(\C)+1$
for each nontrivial linear code $\C$ (see Theorem \ref{thm-loca-ext-cycliccode}), and settled the minimum linear locality of $(\overline{\C})^\perp$ for each nontrivial binary linear code $\C$ (see Corollary \ref{cor-localityextedbinarycode}).
\item We proved that the minimum linear locality of an NMDS code $\C$ is either $d(\C^\perp)-1$ or  $d(\C^\perp)$ (see Theorem \ref{thm-myJ2121}), and further proved that $\C$ is either a $d$-optimal and $k$-optimal or an almost $d$-optimal and
$k$-optimal LLRC (see Theorem \ref{thm-nmdslocality0}).
\end{enumerate}
These general results have settled the minimum linear locality of many families of nontrivial linear codes. Furthermore, the minimum linear locality of these codes was not studied in the literature.  Hence,
we have reached our first objective.

After studying a number of families of known linear codes with the general theory developed, we have found many classes of  optimal LLRCs. These optimal LLRCs were not reported in the literature. Table \ref{tab-1} lists fourteen classes of $k$-optimal LLRCs. Table \ref{tab-2}
lists nineteen classes of LLRCs which are both $d$-optimal and $k$-optimal and have different parameters.  In both tables,
  \begin{itemize}
    \item $n_h=(q^m-1)/(q-1)$,
    \item $\surd$ means that the code is optimal with the Singleton-like or CM bound,
    \item $A$ means that the code is almost optimal with respect to the Singleton-like bound, and
    \item $?$ means that the optimality is open.
  \end{itemize}
These classes of optimal LLRCs demonstrate that we have reached our second objective.

We remark that the locality of locally recoverable codes in the literature actually is the linear locality, but may not be the
minimum linear locality. This paper has treated the minimum linear locality of nontrivial linear codes.
Availability is another interesting parameter of LLRCs. All the LLRCs presented in this paper naturally have availability $1$,
and some of them may have availability $2$ or more. It is extremely hard to
develop general theory for LLRCs with availability more than $1$, although such LLRCs are available in the literature. To study the maximum availability of an LLRC code $\C$ with respect to the minimum linear locality $d(\C^\perp)-1$,
one has to characterize the minimum weight codewords in $\C^\perp$. This is a very hard problem in general. The reader
is cordially invited to investigate the maximum availability of these optimal LLRCs documented in this paper.
Finally, we point out that all the LLRCs presented in this paper are from known linear codes in the literature and our objective
is to study their minimum linear locality and optimality with respect to the Singleton-like and CM bounds.

\begin{table}
  \centering
  \caption{Some $k$-optimal LLRCs  from known codes}\label{tab-1}
  \begin{tabular}{|c|c|c|c|c|c|c|}
    \hline
    $\C$ & $n$ & $k$ & $d$ & $r$ & $d_{opt}$ & $k_{opt}$ \\ \hline
    $\cH_{(q,m)}$ & $n_h$ & $n_h-m$ & $3$ & $q^{m-1}-1$ & $?$ & $\surd$ \\ \hline
    $\Sim_{(q,m)}$ & $n_h$ & $m$ & $q^{m-1}$ & $2$ & $?$ & $\surd$ \\ \hline
    $(\cH_{(q,m)})_{\{t_1\}}$ & $n_h-1$ & $n_h-m-1$ & $3$ & $q^{m-1}-2$ & $?$ & $\surd$ \\ \hline
    $((\cH_{(q,m)})_{\{t_1\}})^\perp$ & $n_h-1$ & $m$ & $q^{m-1}-1$ & $2$ & $?$ & $\surd$ \\ \hline
        $(\Sim_{(q,m)})_{\{t_1\}}$ & $n_h-1$ & $m-1$ & $q^{m-1}$ & $1$ & $?$ & $\surd$ \\ \hline
        $((\Sim_{(q,m)})_{\{t_1\}})^\perp$ & $n_h-1$ & $n_h-m$ & $2$ & $q^{m-1}-1$ & $?$ & $\surd$ \\ \hline
    $\R_q(1,m)$ & $q^m$ & $m+1$& $(q-1)q^{m-1}$ & $2$ & $?$ & $\surd$ \\ \hline
    $\R_q(1,m)^\perp$ & $q^m$ & $q^m-m-1$& $3$ & $q^{m}-q^{m-1}-1$ & $?$ & $\surd$ \\ \hline
  $\C_f$ (Thm. \ref{thm-bestexamthm}) & $2^m+1$ & $3$ & $2^m-2$ & $3$ & $A$ & $\surd$ \\   \hline
  $\bar{\C}_f$ (Thm. \ref{thm-bestexamthm2}) & $2^m+2$ & $3$ & $2^m-1$ & $3$ & $A$ & $\surd$ \\   \hline
      $\C_o$ & $q^2+1$ & $4$  & $q^2-q$ & $3$ & $?$ & $\surd$ \\ \hline
      $(\C_o)_{\{t_1\}}$ & $q^2$ & $3$  & $q^2-q$ & $2$ & $?$ & $\surd$ \\ \hline
      $(\C_o)^{\{t_1\}}$ & $q^2$ & $4$  & $q^2-q-1$ & $3$ & $?$ & $\surd$ \\ \hline
     $\C(\cA)$ (Thm. \ref{thm-localitymaxarc}) & $n$ & $3$ & $n-h$ & $2$ & $?$ & $\surd$    \\
    \hline
  \end{tabular}
\end{table}

\begin{table}
  \centering\small
  \caption{Both $d$-optimal and $k$-optimal LLRCs from known codes}\label{tab-2}
  \begin{tabular}{|c|c|c|c|c|c|c|}
    \hline
    $\C$ & $n$ & $k$ & $d$ & $r$ & $d_{opt}$ & $k_{opt}$ \\ \hline
    $\cH_{(q,3)}$ & $q^2+q+1$ & $q^2+q-2$ & $3$ & $q^2-1$ & $\surd$ & $\surd$ \\ \hline
    $(\cH_{(q,3)})_{\{t_1\}}$ & $q^2+q$ & $q^2+q-3$ & $3$ & $q^2-2$ & $\surd$ & $\surd$ \\ \hline
    $((\Sim_{(q,3)})_{\{t_1\}})^\perp$ & $q^2+q$ & $q^2+q-2$ & $2$ & $q^2-1$ & $\surd$ & $\surd$ \\ \hline
    $\C_o^\perp$ & $q^2+1$ & $q^2-3$ & $4$ & $q^2-q-1$ & $\surd$ & $\surd$ \\ \hline
    $(\C_o^\perp)_{\{t_1\}}$ & $q^2$ & $q^2-4$ & $4$ & $q^2-q-2$ & $\surd$ & $\surd$ \\ \hline
    $(\C_o^\perp)^{\{t_1\}}$ & $q^2$ & $q^2-3$ & $3$ & $q^2-q-1$ & $\surd$ & $\surd$ \\ \hline
      $\C(\cA)^\perp$ (Thm. \ref{thm-localitymaxarc}) & $n$ & $n-3$ & $3$ & $n-h-1$ & $\surd$ & $\surd$ \\ \hline
    $\C_{(3^s,3^s+1,3,1)} $ & $3^s+1$ & $3^s-3$ & $4$ & $3^s-4$ & $\surd$ & $\surd$ \\ \hline
   $\C_{(3^s,3^s+1,3,1)}^\perp$ & $3^s+1$ & $4$ & $3^s-3$ & $3$ & $\surd$ & $\surd$ \\ \hline
    $\C_{(2^s,2^s+1,3,1)}$ & $2^s+1$ & $2^s-3$ & $4$ & $2^s-4$ & $\surd$ & $\surd$ \\ \hline
    $\C_{(2^s,2^s+1,3,1)}^\perp$ & $2^s+1$ & $4$ & $2^s-3$ & $3$ & $\surd$ & $\surd$ \\  \hline
    $\C_{(2^s,2^s+1,4,1)}$ & $2^s+1$ & $2^s-5$ & $6$ & $2^s-6$ & $\surd$ & $\surd$ \\ \hline
    $\C_{(2^s,2^s+1,4,1)}^\perp$ & $2^s+1$ & $6$ & $2^s-5$ & $5$ & $\surd$ & $\surd$ \\  \hline
  $\C_f^\perp$ (Thm. \ref{thm-bestexamthm}) & $2^m+1$ & $2^m-2$ & $3$ & $2^m-3$ & $\surd$ & $\surd$ \\ \hline
  $\bar{\C}_f^\perp$ (Thm. \ref{thm-bestexamthm2}) & $2^m+2$ & $2^m-1$ & $3$ & $2^m-2$ & $\surd$ & $\surd$ \\  \hline
  $\bar{\C}_f^\perp$ (Thm.  \ref{thm-21021-211}) & $2^m+3$ & $2^m$ & $3$ & $2^m-1$ & $\surd$ & $\surd$ \\   \hline
   $\bar{\C}_f$ (Thm.  \ref{thm-2345}) & $2^m+3$ & $3$ & $2^m$ & $2$ & $\surd$ & $\surd$ \\  \hline
    $(\overline{\C_{(2^s,2^s+1,3,1)}})^\perp$ (Thm \ref{thm-2SQScodeext}) & $2^s+2$ & $5$ & $2^s-3$ & $4$ & $\surd$ & $\surd$ \\  \hline
    $(\overline{\C_{(3^s,3^s+1,3,1)}})^\perp$ (Thm \ref{thm-SQScodeext}) & $3^s+2$ & $5$ & $3^s-3$ & $4$ & $\surd$ & $\surd$ \\
    \hline
  \end{tabular}
\end{table}

\end{document}